\documentclass[journal,twoside,web]{ieeecolor}
\usepackage{generic}
\usepackage{cite}
\usepackage{amsmath,amssymb,amsfonts}
\usepackage{graphicx}
\usepackage{hyperref}
\usepackage{algpseudocode}
\usepackage{textcomp}
\usepackage{xcolor}
\usepackage{stfloats}
\usepackage{booktabs}
\let\labelindent\relax
\usepackage{enumitem}
\usepackage{arydshln}
\usepackage{textcomp}
\usepackage{multirow}
\usepackage{mathtools}
\usepackage[thinc]{esdiff}
\usepackage{times}
\usepackage{bm}
\usepackage{algorithm}
\usepackage{algpseudocode}
\usepackage{accents}

\usepackage{tikz}

\hypersetup{hypertex=true,
	colorlinks=true,
	linkcolor=blue,
	urlcolor=black,
	anchorcolor=blue,
	citecolor=blue}

\newlength{\dhatheight}
\newcommand{\doublehat}[1]{%
	\settoheight{\dhatheight}{\ensuremath{\hat{#1}}}%
	\addtolength{\dhatheight}{-0.35ex}%
	\hat{\vphantom{\rule{1pt}{\dhatheight}}%
		\smash{\hat{#1}}}}

\newtheorem{assumption}{Assumption}[section]
\newtheorem{definition}{Definition}[section]

\newtheorem{theorem}{Theorem}[section]

\newtheorem{lemma}{Lemma}

\newtheorem{remark}{Remark}

\newcommand\norm[1]{\left\lVert#1\right\rVert}

\def\BibTeX{{\rm B\kern-.05em{\sc i\kern-.025em b}\kern-.08em
		T\kern-.1667em\lower.7ex\hbox{E}\kern-.125emX}}
\begin{document}
	\title{Range Space or Null Space: Least-Squares Methods for the Realization Problem}
	\author{Jiabao He, Yueyue Xu, Yue Ju, Cristian R. Rojas and H\r{a}kan Hjalmarsson
		\thanks{Jiabao He, Yue Ju, Cristian R. Rojas and H\r{a}kan Hjalmarsson are with the Division of Decision and Control Systems, School of Electrical Engineering and Computer Science, KTH Royal Institute of Technology, 100 44 Stockholm, Sweden.  (Emails: jiabaoh, yuej, crro, hjalmars@kth.se)}
		\thanks{Yueyue Xu is with Department of Mathematics, KTH Royal Institute of Technology, 100 44 Stockholm, Sweden. (Emails: yueyuex@kth.se)}
		\thanks{This work was supported by VINNOVA Competence Center AdBIOPRO, contract [2016-05181] and by the Swedish Research Council through the research environment NewLEADS (New Directions in Learning Dynamical Systems), contract [2016-06079], and contract 2019-04956.}}
	
	\maketitle
	
	\begin{abstract}
		This contribution revisits the classical approximate realization problem, which involves determining matrices of a state-space model based on estimates of a truncated series of Markov parameters. A Hankel matrix built up by these Markov parameters plays a fundamental role in this problem, leveraging the fact that both its range space and left null space encode critical information about the state-space model. We examine two prototype realization algorithms based on the Hankel matrix: the classical range-space-based (SVD-based) method and the more recent null-space-based method. It is demonstrated that the range-space-based method corresponds to a total least-squares solution, whereas the null-space-based method corresponds to an ordinary least-squares solution. By analyzing the differences in sensitivity of the two algorithms, we determine the conditions when one or the other realization algorithm is to be preferred, and identify factors that contribute to an ill-conditioned realization problem. Furthermore, recognizing that both methods are suboptimal, we argue that the optimal realization is obtained through a weighted least-squares approach. A statistical analysis of these methods, including their consistency and asymptotic normality is also provided.
	\end{abstract}
	
	
	\section{Introduction} \label{Sct1}
	
	Consider the following discrete-time linear time-invariant (LTI) system:
	\begin{subequations} \label{E1}
		\begin{align}
			x_{k + 1} &= Ax_{k} + Bu_{k}, \label{E1a}\\
			y_{k} &= Cx_{k}, \label{E1b}		
		\end{align}
	\end{subequations}
	where $x_{k}\in \mathbb{R}^{n_x}$, $u_{k}\in \mathbb{R}^{n_u}$ and $y_{k}\in \mathbb{R}^{n_y}$ are the system state, input and output, respectively. We make the following assumption:
	
	\begin{assumption} [System]\label{Asp1}
		System \eqref{E1} is stable and minimal, i.e., the spectral radius of $A$ satisfies $\rho(A) < 1$, and $(A,B)$ is controllable and $(A,C)$ is observable.
	\end{assumption}
	
	This work revisits the classical approximate realization problem: given estimates of the first $n$ Markov parameters $\left\{g_i = CA^{i}B\right\}$, where $i = 0,1,\cdots,n-1$, we aim to determine a realization of system matrices $A\in \mathbb{R}^{n_x\times n_x}$, $B\in \mathbb{R}^{n_x\times n_u}$ and $C\in \mathbb{R}^{n_y\times n_x}$, up to a similarity transform. To simplify the presentation, the main body of this work focuses on single-input-single-output (SISO) systems, i.e., $n_x=n_y=1$. Extensions to multi-input-multi-output (MIMO) systems are discussed in Section \ref{Sct6}.
	
	The following Hankel matrix plays a fundamental role in the realization problem\footnote{While other matrix structures, such as the Page matrix \cite{Van1983approximate}, can also be used to obtain a realization of the system matrices, we focus exclusively on the Hankel matrix.}:
	\begin{equation} \label{E2}
		\mathcal{H}_{fp} = \begin{bmatrix}
			g_0 & g_1 & \cdots & g_p \\
			g_1 & g_2 & \cdots & g_{p+1} \\
			\vdots & \vdots & \ddots & \vdots \\
			g_f & g_{f+1} & \cdots & g_{n-1}
		\end{bmatrix}\in \mathbb{R}^{(f+1)\times(p+1)},
	\end{equation}
	where the number of row and column\footnote{In the literature of subspace identification, the number of row and column refer to the future horizon and past horizon, respectively.} satisfy $f+p+1=n$, and $f,p \geq n_x$. It is well known that $\mathcal{H}_{fp}$ can be factorized as
	\begin{equation} \label{E3}
		\mathcal{H}_{fp} = \mathcal{O}_{f}\mathcal{C}_p,
	\end{equation}
	where
	\begin{subequations} \label{E4} 
		\begin{align}
			\mathcal{O}_{f} &:= \begin{bmatrix}
				{{C^{\top}}}&{{{\left({CA} \right)}^{\top}}}& \cdots &{{{\left({C{A^{{f}}}}\right)}^{\top}}}
			\end{bmatrix}^{\top}, \\
		    \mathcal{C}_p &:= \begin{bmatrix}B&{AB}& \cdots &{{A^{p}}B}\end{bmatrix},
		\end{align}
	\end{subequations}
    are the extended observability matrix and controllability matrix, respectively. Under Assumption \ref{Asp1}, we have that $\text{rank}\left(\mathcal{O}_{f}\right) = \text{rank}\left(\mathcal{C}_{p}\right) = n_x$. Therefore, the Hankel matrix $\mathcal{H}_{fp}$ has rank equal to $n_x$ \cite{Kailath1980linear}. Let us now assume that the order $n_x$ has been correctly selected, and our main focus is to obtain a realization of system matrices. It turns out that there are still some facets not previously recognized in the literature. 
	
	Most realization algorithms are based on subspaces of $\mathcal{H}_{fp}$, as these subspaces encode key information regarding the underlying state-space model. The range space of $\mathcal{H}_{fp}$, denoted by $\mathcal{R}(\mathcal{H}_{fp})$, is the set of all linear combinations of the columns of $\mathcal{H}_{fp}$, and the null (kernel) space of $\mathcal{H}_{fp}$, denoted by $\mathcal{K}(\mathcal{H}_{fp})$, is the set of all solutions of the homogeneous equation $\mathcal{H}_{fp}\alpha=0$. Similar definitions apply to the row space $\mathcal{R}(\mathcal{H}_{fp}^\top)$ and the left null space $\mathcal{K}(\mathcal{H}_{fp}^\top)$. In terms of dimensionality, the range space $\mathcal{R}(\mathcal{H}_{fp})$ and the row space $\mathcal{R}(\mathcal{H}_{fp}^\top)$ both have dimension $n_x$, equal to the rank of $\mathcal{H}_{fp}$; the null space $\mathcal{K}(\mathcal{H}_{fp})$ has dimension $p+1-n_x$, and the left null space $\mathcal{K}(\mathcal{H}_{fp}^\top)$ has dimension $f+1-n_x$. Moreover, concerning orthogonality, $\mathcal{R}(\mathcal{H}_{fp})$ and $\mathcal{K}(\mathcal{H}_{fp}^\top)$ are orthogonal complements, as are $\mathcal{R}(\mathcal{H}_{fp}^\top)$ and $\mathcal{K}(\mathcal{H}_{fp})$.
	According to \eqref{E3} we further have
	\begin{subequations} \label{E6} 
		\begin{align}
			\mathcal{R}(\mathcal{H}_{fp}) & = \mathcal{R}(\mathcal{O}_{f}), \ \mathcal{K}(\mathcal{H}_{fp}^\top) = \mathcal{K}(\mathcal{O}_{f}^\top), \\
			\mathcal{R}(\mathcal{H}_{fp}^\top) & = \mathcal{R}(\mathcal{C}_{p}^\top), \ \mathcal{K}(\mathcal{H}_{fp}) = \mathcal{K}(\mathcal{C}_{p}).		
		\end{align}
	\end{subequations}
	As shown in \eqref{E6}, there is a symmetric structure in the four subspaces of $\mathcal{H}_{fp}$. For the purpose of finding a realization, it is sufficient to study one pair of the range and null subspaces of $\mathcal{H}_{fp}$. As with in most realization algorithms, we focus on the range space and left null space of the extended observability matrix $\mathcal{O}_{f}$, or equivalently, $\mathcal{R}(\mathcal{H}_{fp})$ and $\mathcal{K}(\mathcal{H}_{fp}^\top)$.
	
	The range space was studied first, as in the celebrated Ho-Kalman algorithm \cite{Ho1966effective}, Kung's algorithm \cite{Kung1978new}, the eigensystem realization algorithm (ERA) \cite{Juang1985eigensystem} and many variants of subspace identification methods (SIMs) \cite{Larimore1990canonical,Van1994n4sid,Verahegen1992subspace,Qin2005novel,Jansson2003subspace,Chiuso2007role}. For convenience, we refer to this type of realization methods as the range-space-based realization (RASBR). Our presentation of RASBR is primarily based on the prototype algorithm introduced by Kung \cite{Kung1978new}. Although this method is outdated, as more powerful SIMs have later been proposed in the literature, it captures some important ideas that characterize RASBR algorithms, and is convenient for illustration purposes. Moreover, the future horizon is fixed at $f=n_x$. The impact of $f$ is discussed in Section \ref{Sct3}. 
	
	\textbf{Range-Space-Based Realization (RASBR)}: The starting point of RASBR is the SVD of $\mathcal{H}_{{n_x}p}$:
	\begin{equation} \label{E7} 
		\mathcal{H}_{{n_x}p} = USV^{\top} = U_1S_1V_1^{\top}\in \mathbb{R}^{(n_x+1)\times(p+1)},
	\end{equation}
	where 
	\begin{equation*}
		\begin{split}
			U & = \begin{bmatrix}u_1&u_2&\cdots&u_{n_x+1}\end{bmatrix}\in \mathbb{R}^{(n_x+1)\times(n_x+1)}, \\
			V & = \begin{bmatrix}v_1&v_2&\cdots&v_{p}\end{bmatrix}\in \mathbb{R}^{(p+1)\times(p+1)}, \\
			S & = \begin{bmatrix}
				\text{diag}\left(\sigma_1,\sigma_2,\cdots,\sigma_{n_x},0\right)&0
			\end{bmatrix}\in \mathbb{R}^{(n_x+1)\times(p+1)},
		\end{split}
	\end{equation*}
	and $u_i \in \mathbb{R}^{n_x+1}$ with $v_i \in \mathbb{R}^{p+1}$ are left and right singular vectors, respectively. Moreover, $S_1 = \text{diag}\left(\sigma_1,\sigma_2,\dots,\sigma_{n_x}\right)$ is a diagonal matrix, having the first $n_x$ non-zero singular values of $S$ on its diagonal and satisfying $\sigma_1 \geq \sigma_2 \geq \cdots \geq \sigma_{n_x}$, and $U_1\in \mathbb{R}^{(n_x+1)\times n_x}$ with $V_1\in \mathbb{R}^{(p+1)\times n_x}$ contain the left and right singular vectors corresponding to $S_1$, respectively. From \eqref{E7}, balanced realizations of $\mathcal{O}_{n_x}$ and $\mathcal{C}_p $ are given by
	\begin{subequations} \label{E8} 
		\begin{align}
		    \bar{\mathcal{O}}_{n_x} &= U_1S_1^{1/2}, \\
		    \bar{\mathcal{C}}_p &= S_1^{1/2} V_1^{\top}.		
	    \end{align}
	\end{subequations}
	Moreover, based on the shift-invariant property of $\bar{\mathcal{O}}_{n_x}$, the system matrices are given by
	\begin{subequations} \label{E9} 
		\begin{align}
			{\bar C}_{R} &= \bar{\mathcal{O}}_{n_x}(1,:), \\
			{\bar A}_{R} &=(\bar{\mathcal{O}}_{n_x}^{+})^{\dagger}\bar{\mathcal{O}}_{n_x}^{-},\\
			{\bar B}_{R} &= \bar{\mathcal{C}}_p (:,1),
		\end{align}
	\end{subequations}
	where $\bar{\mathcal{O}}_{n_x}^{-}$ and $\bar{\mathcal{O}}_{n_x}^{+}$ are the last and first $n_x$ rows of $\bar{\mathcal{O}}_{n_x}$, respectively. The notation $(\cdot)^{\dagger}$ is the Moore-Penrose pseudo-
	inverse of a matrix. Moreover, the indexing of matrices, such as $\bar{\mathcal{O}}_{n_x}(1,:)$ and $\bar{\mathcal{C}}_p (:,1)$, follow MATLAB syntax. 
	
	\textbf{Null-Space-Based Realization (NUSBR)}: Unlike the considerable attention given to RASBR, NUSBR methods have received comparatively less focus and been somewhat overlooked. Although earlier work \cite{Ottersten1994subspace, Viberg1997analysis,Swindlehust1995subspace,Jansson1996linear} recognized that the null space of the extended observability matrix $\mathcal{O}_{n_x}$ could be obtained via two step weighted least-squares, their methods require an explicit estimate of $\mathcal{O}_{n_x}$ to construct the optimal weighting, necessitating the use of SVD to obtain such an estimate. Moreover, given a set of so-called Kronecker indices (observability indices), an echelon state-space realization was given in \cite{Hannan2012statistical}. This type of realization establishes a relation between the coefficients of the system's characteristic polynomial and the left null space of the Hankel matrix. It was further argued in \cite[Chapter 3]{Regalia2018adaptive} that virtually all identification methods reduce to finding a vector in the null space of a Hankel matrix built from data from the system to identify. Building on this body of work, we recently proposed a novel NUSBR method \cite{He2024weighted} which bypasses the estimate of $\mathcal{O}_{n_x}$ and directly estimates its left null space from the Hankel matrix $\mathcal{H}_{n_xp}$ using least-squares. To illustrate it, according to the Cayley-Hamilton theorem, we have that
	\begin{equation} \label{E10} 
		A^{n_x} + a_1 A^{{n_x}-1} + \cdots + a_{{n_x}-1} A + a_{{n_x}}I = 0,
	\end{equation}
	where $\left\{a_i\right\}_{i=1}^{n_x}$ are coefficients of the characteristic polynomial of matrix $A$. Since $\text{rank}\left({\mathcal{O}}_{n_x}\right) = n_x$ and ${\mathcal{O}}_{n_x}\in \mathbb{R}^{(n_x+1)\times n_x}$, we have that the dimension of the left null space of ${\mathcal{O}}_{n_x}$ is
	\begin{equation} \label{E10a} 
		\text{dim}\left(\mathcal{K}({\mathcal{O}}_{n_x}^\top)\right) = 1.
	\end{equation}
	Using equation \eqref{E10}, we have that
	\begin{equation}  \label{E11} 
		\begin{bmatrix}
			a_{n_x}&a_{{n_x}-1}& \cdots &a_1&1\end{bmatrix} {\mathcal{O}}_{n_x} = 0.
	\end{equation}
	From \eqref{E10a} and \eqref{E11} it follows that the left null space of ${\mathcal{O}}_{n_x}$ is completely parameterized by the coefficients $\left\{a_i\right\}_{i=1}^{n_x}$. According to \eqref{E6}, $\mathcal{K}(\mathcal{H}_{n_xp}^\top) = \mathcal{K}(\mathcal{O}_{n_x}^\top)$, i.e., the left null space of $\mathcal{O}_{n_x}$ is also the left null space of $H_{{n_x}p}$. Define
	\begin{equation} \label{E12}
		\bm{a} := \begin{bmatrix}
			a_{n_x}&a_{{n_x}-1}& \cdots &a_1\end{bmatrix}.
	\end{equation}
	We then have $\begin{bmatrix}\bm{a}&1\end{bmatrix}\mathcal{H}_{{n_x}p} = 0$. Correspondingly, we partition $\mathcal{H}_{{n_x}p}$ into two parts:
	\begin{equation} \label{E13}
		\mathcal{H}_{{n_x}p} =  \begin{bmatrix}
			{{g_0}}&{{g_1}}& \cdots &{{g_p}}\\
			{{g_1}}&{{g_2}}& \cdots &{{g_{p + 1}}}\\
			\vdots & \vdots & \ddots & \vdots \\
			\hline
			{g_{n_x}}&{g_{n_x + 1}}& \cdots &{{g_{n-1}}}
		\end{bmatrix} =: \begin{bmatrix}{\mathcal{H}}_{n_xp}^{+}\\ \hline
			{\mathcal{H}}_{n_xp}^{-}\end{bmatrix},
	\end{equation}
	which further gives
	\begin{equation} \label{E14}
		\bm{a} \mathcal{H}_{n_xp}^{+} + \mathcal{H}_{n_xp}^{-} = 0.
	\end{equation}
	In this way, $\bm{a}$ is given by
	\begin{equation} \label{E15}
		{\bm{a}} = -{\mathcal{H}}_{n_xp}^{-}({ \mathcal{H}}_{n_xp}^{+})^{\dagger}.
	\end{equation}
	Then, system matrices $A$, $B$ and $C$ in their observer canonical form \cite{Kailath1980linear} are given as follows: 
	\begin{subequations} \label{E16}
		\begin{align}
			{\bar C}_{N} &= \begin{bmatrix}1&0&0&\cdots&0
			\end{bmatrix}, \\
			{\bar A}_{N} &= \begin{bmatrix}
				-{a}_{1}&1&0&\cdots&0\\
				-{a}_{2}&0&1&\cdots&0\\
				\vdots&\vdots&\vdots&\ddots&\vdots\\
				-{a}_{n_x}&0&0&\cdots&0
			\end{bmatrix}, \\
			{\bar B}_{N} &= \bar{\mathcal{O}}_{n-1}^{\dagger}\begin{bmatrix}
				g_{0}&g_{1}&\cdots&g_{n-1}
			\end{bmatrix}^\top,
		\end{align}
	\end{subequations}
	where the extended observability matrix $\bar{\mathcal{O}}_{n-1}$ is computed from matrices ${\bar C}_{N}$ and ${\bar A}_{N}$.
	
	Clearly, if the first $n$ Markov parameters $\left\{g_i = CA^{i}B\right\}_{i=0}^{n-1}$ are the true values, then the above RASBR and NUSBR algorithms return the true system matrices, up to a similarity transformation. Thus, under these ideal conditions, the choice of method does not matter. In practice, however, we are typically given estimates of these Markov parameters, usually obtained via least-squares methods from a finite impulse response (FIR) model or an auto-regressive exogenous input (ARX) model, in which case the estimates of these two algorithms may be very different. We illustrate this on a simple example.
	
	\textbf{Experiment 1}: Consider system matrices
	\begin{equation} \label{E17}
		A = \begin{bmatrix}\lambda&\delta\\0&\lambda\end{bmatrix},
		B = \begin{bmatrix}0\\1\end{bmatrix}, 
		C = \begin{bmatrix}1&0\end{bmatrix}.
	\end{equation}
	System 1 is given by taking $\lambda = 0.1, \delta=2$, and System 2 is given by taking $\lambda = 0.9, \delta=10$. It can be verified that the above two systems satisfy Assumption \ref{Asp1}. We add independent Gaussian noise with unit variance to the first 20 Markov parameters of these two systems, and use the RASBR and NUSBR algorithms to obtain realizations of system matrices. We measure the estimation quality using the so-called FIT, defined as
	\begin{equation} \label{E17a}
		{\rm{FIT}} = 100\times\left(1-\frac{\norm{g_{\circ} - \hat g}}{\norm{g_{\circ} - {\bar g}_{\circ}}}\right),
	\end{equation}
	where $g_{\circ}$ includes the first 100 Markov parameters of the true system, ${\bar g}_{\circ}$ is the average of $g_{\circ}$, $\hat g$ includes the Markov parameters of the estimated systems, and $\norm{\cdot}$ is the Euclidean norm. We run 200 Monto Carlo trails, and the results are shown in Figure \ref{F2}. As observed from the figure, for System 1, NUSBR performs better than RASBR, whereas for System 2, RASBR is better. This indicates that the accuracy of the two realization algorithms is both case-specific and suboptimal-- a phenomenon that has not been previously recognized in the literature. This discrepancy raises at least the following questions:
	
	(1) Why do the two realization methods perform differently, and when does one outperform the other?
	
	(2) Is there an optimal realization algorithm, and, if so, what is it?
	
	(3) What are the statistical properties of these algorithms?
	
	The main motivation of this work is to address the above questions. 
	\begin{figure}
		\centering
		\includegraphics[scale=0.5]{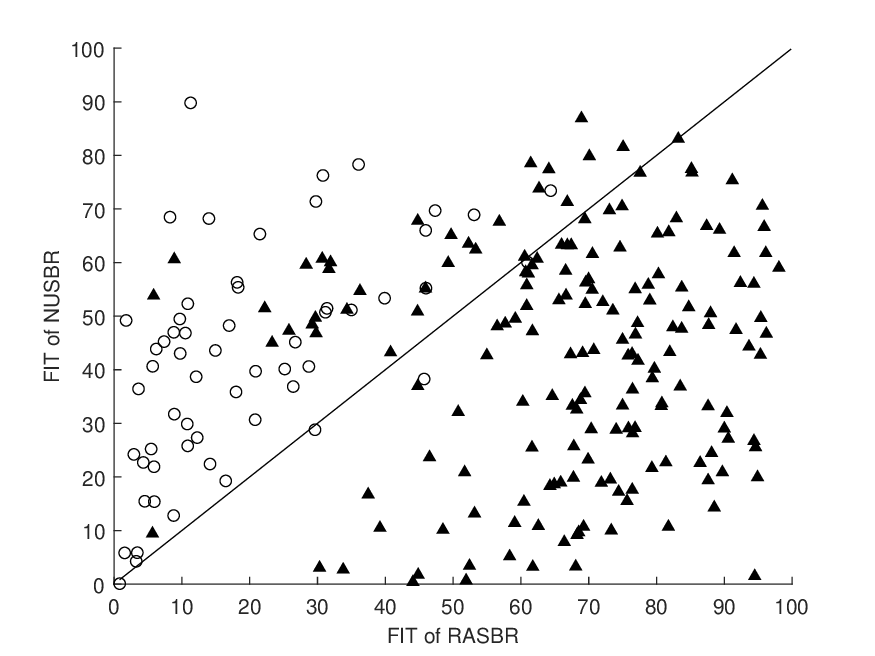}
		\caption{FITs of impulse responses from 200 Monte Carlo trials: System 1 ($\circ$) and System 2 ($\blacktriangle$), and the solid line is a bisector line.}
		\label{F2}
	\end{figure}

	\subsection{Contribution} \label{Sct1.2}
	
	The contribution of this paper are three-fold:
	
	(1) We compare the above two realization methods within a least-squares framework, where it is revealed that RASBR corresponds to a total least-squares (TLS) solution, and NUSBR corresponds to an ordinary least-squares (OLS) solution. Moreover, by analyzing the differences in sensitivity of the two algorithms, we determine the conditions when one or the other realization algorithm is to be preferred, and identify the factors that contribute to an ill-conditioned problem.
	
	(2) We argue that the optimal realization is a weighted null space fitting problem, and provide a weighted least-squares (WLS) solution for solving it.
	
	(3) We provide a statistical analysis for these methods, including consistency and asymptotic normality. Our results suggest that under mild assumptions, the OLS and TLS solutions are asymptotically equivalent, while the WLS solution gives the smallest asymptotic variance. 
	
	Although the realization algorithms considered in this paper are simplified, they serve as prototype algorithms for state-of-the-art SIMs. Therefore, these findings provide a perspective to explain why the performance of many SIMs tends to be suboptimal and case-dependent. Moreover, they shed new light on pursuing an asymptotically efficient SIM, a long-standing open problem in system identification.
	
	\subsection{Related Work} \label{Sct1.3}
	
	The foundational work on the approximate realization problem is the celebrated Ho-Kalman algorithm \cite{Ho1966effective}, which is computationally efficient and straightforward to implement. However, this algorithm operates under the assumption of noise-free data. When applied to an Hankel matrix using corrupted Markov parameters, the resulting matrix often has an incorrect rank and can lead to poor performance \cite{Van1983approximate}. Subsequent algorithms, such as Kung's method \cite{Kung1978new} and the ERA algorithm \cite{Juang1985eigensystem}, were developed as modifications of the original Ho-Kalman algorithm, where SVD remains as a central step. These advancements paved the way for modern SIMs \cite{Larimore1990canonical,Van1994n4sid,Verahegen1992subspace,Qin2005novel,Jansson2003subspace,Chiuso2007role}. A key innovation in SIMs is the direct estimation of the Hankel (Hankel-like) matrix $\mathcal{H}_{fp}$ in a unstructured manner, rather than deriving it from a series of Markov parameters. Additionally, various data-dependent weighting matrices are applied to the estimated Hankel matrix prior to the SVD step to further enhance performance.
	
	Although most SVD-based realization algorithms focus on the range space of the observability matrix, earlier work \cite{Ottersten1994subspace, Viberg1997analysis,Swindlehust1995subspace,Jansson1996linear} have recognized the potential of utilizing the null space for obtaining an optimal estimate of the system's poles. Moreover, it was highlighted in \cite{De2019least, De2020least} that the least-squares optimal realization of autonomous LTI systems can be reformulated as a multi-parameter eigenvalue problem. This problem can be solved by applying forward shift recursions to a given set of multivariate polynomial equations, generating so-called block Macaulay matrices. A key concept therein is the elimination of the state vector by leveraging the Cayley-Hamilton theorem, with similar ideas also discussed in \cite{Nicolai2023realizing,Bakshi2023new}. Furthermore, a weighted null-space fitting (WNSF) method was proposed in \cite{Galrinho2018parametric}, which gives an asymptotically efficient estimate for many common parametric model structures, such as output-error, ARMAX and Box-Jenkins. Very recently, we extended this approach to state-space models \cite{He2024weighted}.
	
	Given the close relationship between SVD and the TLS problem, the approximate realization problem was treated as a special global TLS problem in \cite{Markovsky2005application}, where a kernel representation of the system is used. The problem is then related to a structured TLS problem, and subsequently solved using a solution developed for the structured TLS problem. Related studies can be found in \cite{De1993singular,De1993structured,Markovsky2007overview,De2002total}.
	
	In summary, although many solutions have been proposed for the classical realization problem, it is widely acknowledged that many of these methods are often suboptimal and highly case-dependent, with the underlying reasons remaining unclear. Additionally, there has been no systematic effort to compare these methods within a unified framework. This work aims to address this gap using a least-squares framework.
	
	\subsection{Structure} \label{Sct1.4}
	
	The disposition of the paper is as follows: We briefly present some preliminaries of TLS and OLS problems in Section~\ref{Sct2}. In Section~\ref{Sct3}, we first cast the RASBR and NUSBR algorithms introduced in this section into a TLS and OLS problem, respectively. Under the least-squares framework, we then analyze the sensitivity of the solutions for TLS and OLS, which determines the conditions when one or the other method is to be preferred. In Section~\ref{Sct4}, we argue that the optimal realization is a WLS solution. In Section~\ref{Sct5}, statistical properties of the three least-squares methods, including consistency and asymptotic normality, are presented. Finally, the paper is concluded in Section~\ref{Sct6}. All proofs and technical lemmas are provided in the Appendix.
	
	\subsection{Notations}  \label{Sct1.5}
	(1) For a matrix $X$ with appropriate dimensions, $X^\top$, $X^{-1}$, $X^{1/2}$, $X^\dagger$, $\norm{X}$, $\norm{X}_F$, ${\rm{rank}}(X)$ and $\rho(X)$ denote its transpose, inverse, square root, Moore$\mbox{-}$Penrose pseudo-inverse, spectral norm, Frobenius norm, rank, and spectral radius, respectively. $X \succ(\succcurlyeq)$ $0$ means that $X$ is positive (semi) definite. ${\rm{diag}}(X_1,X_2)$ is a block matrix having $X_1$ and $X_2$ on its diagonal. The notation $X_1\otimes X_2$ denotes the Kronecker product of $X_1$ and $X_2$. $\mathcal{R}(X)$ and $\mathcal{K}(X)$ denote the range and null (kernel) subspaces of the matrix $X$, respectively. The matrices $I$ and $0$ are the identity and zero matrices with compatible dimensions.
	 
	(2) The multivariate normal distribution with mean $\mu$ and covariance $\Sigma$ is denoted as $\mathcal{N}(\mu,\Sigma)$. The notation $\mathbb{E}\left\{x\right\}$ is the expectation of a random vector $x$. 
		
	(3) The notation $x_N = \mathcal{O}(f_N)$ means that there is an $M>0$ such that $\limsup\limits_{N \to \infty }\frac{x_N}{f_N} \leq M$.

	\section{Preliminaries on Least-Squares} \label{Sct2}
	
	In this section, we give a brief introduction to the OLS and TLS methods.
	
	\subsection{Ordinary Least-Squares} \label{Sct2.1}
	
	Given a data matrix $\Phi \in \mathbb{R}^{m\times l}$, where $m > l$, and a vector $z\in \mathbb{R}^{m}$ of observations, we seek a vector $\theta\in \mathbb{R}^{l}$ to minimize the following cost function:
	\begin{equation} \label{E18}
		\norm{z - \Phi\theta}.
	\end{equation}
	Following \cite{Golub1980analysis}, the OLS problem can be reformulated as follows:
	\begin{subequations} \label{E19}
		\begin{align}
			\min_{\tilde z} \ &\norm{{\tilde z}}, \label{E19a} \\ 
			\text{s.t.} \ &z+{\tilde z} \in \mathcal{R}(\Phi). \label{E19b}		
		\end{align}
	\end{subequations}
    Once a minimizing ${\tilde z}$ is found, then any $\theta$ satisfying $\Phi{\theta}=z+{\tilde z}$ is called a solution to the OLS problem \eqref{E18}. It is well known that if matrix $\Phi$ has full column rank, the OLS problem then has a unique solution, given by
	\begin{equation} \label{E20}
		{\hat\theta}_\text{ols} = \Phi^\dagger z  = \left(\Phi^\top\Phi\right)^{-1}\Phi^\top z.
	\end{equation}
	There is an underlying assumption for the above problem, namely, all the errors are confined to the observation vector $z$, which is often unrealistic in many scenarios. In practice, the data matrix $\Phi$ is frequently affected by noise as well. The consideration of perturbations of both $z$ and $\Phi$ gives birth to the TLS problem. 
	
	\subsection{Total Least-Squares} \label{Sct2.2}
	
	The TLS problem is given by
	\begin{subequations} \label{E21}
		\begin{align}
			\min_{\tilde z, \tilde \Phi} \ &\norm{\begin{bmatrix}\tilde \Phi&\tilde z\end{bmatrix}}_F, \\
			\text{s.t.} \ &z+\tilde z \in \mathcal{R}(\Phi+\tilde \Phi).
		\end{align}
	\end{subequations}
    Once a minimizing $\begin{bmatrix}\tilde \Phi&\tilde z\end{bmatrix}$ is found, then any ${\theta}$ satisfying $(\Phi+\tilde \Phi){\theta}=z+{\tilde z}$ is called a TLS solution \cite{Van1991total}. The TLS solution is closely related to the SVD. To proceed, let
	\begin{equation} \label{E22}
		\begin{bmatrix}\Phi&z\end{bmatrix} = P\Lambda Q^\top \in \mathbb{R}^{m\times (l+1)},
	\end{equation}
	where
	\begin{equation*}
		\begin{split}
			P & = \begin{bmatrix}p_1&p_2&\cdots&p_m\end{bmatrix}\in \mathbb{R}^{m\times m}, \\
			Q & = \begin{bmatrix}q_1&q_2&\cdots&q_{l+1}\end{bmatrix}\in \mathbb{R}^{(l+1)\times (l+1)}, \\
			\Lambda	& = \begin{bmatrix}
				\text{diag}\left(\sigma_1,\sigma_2,\dots,\sigma_{l+1}\right)\\0
			\end{bmatrix} \in \mathbb{R}^{m\times (l+1)},
		\end{split}
	\end{equation*}
	and $p_i \in \mathbb{R}^{m}$ with $q_i \in \mathbb{R}^{l+1}$ are left and right singular vectors, respectively. Moreover, the singular values satisfy $\sigma_1 \geq \sigma_2 \geq \cdots \geq \sigma_{l+1}$. If $\sigma_{l} > \sigma_{l+1}$, then the unique solution to the above TLS problem is given by
	\begin{equation} \label{E23}
		{\hat\theta}_\text{tls} = -q_{l+1,1:l}q_{l+1,l+1}^{-1},
	\end{equation}
	where $q_{l+1,1:l}$ is a vector containing the first $l$ entries of the vector $q_{l+1}$, and $q_{l+1,l+1}$ is the last entry of $q_{l+1}$. Moreover, the solution \eqref{E23} can be equivalently written as \cite{Golub1980analysis}
	\begin{equation} \label{E24}
		{\hat\theta}_\text{tls} = \left(\Phi^\top\Phi-\sigma_{l+1}^2I\right)^{-1}\Phi^\top z.
	\end{equation}
	From \eqref{E24}, we recognize that the TLS method corrects errors effectively by eliminating the effect of the noises by means of the term $\sigma_{l+1}^2I$. To illustrate this, we define $\bar \Phi = \Phi + \tilde \Phi$, where $\Phi$ is regarded as a deterministic matrix. If $\mathbb{E}\left\{\tilde \Phi\right\} = 0$, we then have $\mathbb{E}\left\{{\bar \Phi}^\top{\bar \Phi}\right\} = \Phi^\top\Phi + \mathbb{E}\left\{{\tilde \Phi}^\top{\tilde \Phi}\right\}$. The underlying assumption herein is that each column of the error matrix $\tilde \Phi$ is independent, and has the same variance $\sigma_{l+1}^2I$, corresponding to the minimal eigenvalue of $\Phi^\top\Phi$. Then, the effect of the noise is effectively corrected by eliminating the term $\sigma_{l+1}^2I$. Further comparisons between OLS and TLS can be found in \cite{Golub1980analysis,Van1991total,Markovsky2007overview}.

	\section{Least-Squares Analysis of Realization Algorithms} \label{Sct3}
	
	In this section, we first cast RASBR and NUSBR algorithms introduced in Section~\ref{Sct1} into two least-squares problems. Under the least-squares framework, we then compare the sensitivity of their solutions and determine  conditions when one or the other realization algorithm is to be preferred.
	
	Compared to matrix $A$ in \eqref{E1}, the estimation of matrices $B$ and $C$ is fairly trivial. Thus, the focus of our work is on the estimation of matrix $A$, which is also at the center of the realization problem. From now on, we use the hat notation $\hat {\mathcal{H}}_{n_xp}$ to represent the Hankel matrix, obtained by replacing $\{g_{i}\}$ by estimates $\left\{\hat g_i\right\}$ in ${\mathcal{H}}_{n_xp}$. Similar conventions apply to other notations.
	
	\subsection{Null-Space-Based Realization is an Ordinary Least-Squares Solution} \label{Sct3.1}
	The reformulation of NUSBR as an OLS problem is quite straightforward. Following \eqref{E18}, after taking $\theta = {\bm{a}}^\top$, $z = -\left({\hat {\mathcal{H}}}_{n_xp}^{-}\right)^\top$, $\Phi = \left(\hat {\mathcal{H}}_{n_xp}^{+}\right)^\top$, and introducing $\tilde z = \left({\tilde {\mathcal{H}}}_{n_xp}^{-}\right)^\top$, the OLS problem can be formulated as in \eqref{E19}:
	\begin{subequations} \label{E25}
		\begin{align}
			\min_{\tilde{\mathcal{H}}_{n_xp}^{-}} \ &\norm{\tilde{\mathcal{H}}_{n_xp}^{-}}, \label{E25a} \\
			\text{s.t.} \ &\left(\hat {\mathcal{H}}_{n_xp}^{-} + \tilde  {\mathcal{H}}_{n_xp}^{-}\right)^\top \in \mathcal{R}\left((\hat {\mathcal{H}}_{n_xp}^{+})^\top\right), \label{E25b}		
		\end{align}
	\end{subequations}
	for which the solution is given by
	\begin{equation} \label{E26}
		\hat{\bm{a}}_\text{ols}^\top = -\left(\hat {\mathcal{H}}_{n_xp}^{+}(\hat {\mathcal{H}}_{n_xp}^{+})^{\top}\right)^{-1}\hat {\mathcal{H}}_{n_xp}^{+}\left({\hat {\mathcal{H}}}_{n_xp}^{-}\right)^\top.
	\end{equation}
	The above solution is identical to the NUSBR method \eqref{E15} with ${\mathcal{H}}_{n_xp}$ being replaced by its approximation $\hat {\mathcal{H}}_{n_xp}$. We therefore conclude that NUSBR is an OLS solution.
	
	\begin{remark} \label{Rmk1}
		As is well known, we can also solve the OLS problem \eqref{E25} using the SVD. Let the SVD of $\hat {\mathcal{H}}_{n_xp}^{+}$ be
		\begin{equation} \label{E26a}
			\hat{\mathcal{H}}_{n_xp}^{+} = \hat U^{+}\hat S^{+} (\hat V^{+})^\top \in \mathbb{R}^{n_x\times(p+1)},
		\end{equation}
		where
		\begin{equation*}
			\begin{split}
				{\hat U}^{+} & = \begin{bmatrix}{\hat u}_1^{+}&{\hat u}_2^{+}&\cdots&{\hat u}_{n_x}^{+}\end{bmatrix}\in \mathbb{R}^{n_x\times n_x}, \\
				{\hat V}^{+} & = \begin{bmatrix}{\hat v}_1^{+}&{\hat v}_2^{+}&\cdots&{\hat v}_{p}^{+}\end{bmatrix}\in \mathbb{R}^{(p+1)\times(p+1)}, \\
				{\hat S}^{+} & = \begin{bmatrix}
					\text{diag}\left({\hat\sigma}_1^{+},{\hat\sigma}_2^{+},\cdots,{\hat\sigma}_{n_x}^{+}\right) & 0
				\end{bmatrix}\in \mathbb{R}^{n_x\times(p+1)}.
			\end{split}
		\end{equation*}
		According to \cite{Golub2013matrix}, the OLS solution $\hat{\bm{a}}_\text{ols}$ is the minimum norm solution of
		\begin{equation} \label{E30a}
			\bm{a}{\hat {\mathcal{H}}}_{n_xp}^{+} = \hat{\bar{\mathcal{H}}}_{n_xp}^{-},
		\end{equation}
		where $\hat{\bar{\mathcal{H}}}_{n_xp}^{-} = \sum_{i=i}^{n_x}\hat v_i^{+}\left({\hat v}_i^{+}\right)^\top \hat {\mathcal{H}}_{n_xp}^{-}$.  
		
	\end{remark}
	
	\subsection{Range-Space-Based Realization is a Total Least-Squares Solution} \label{Sct3.2}
	Using the same notations as in the OLS problem \eqref{E25}, and taking $\tilde  \Phi = \left(\tilde {\mathcal{H}}_{n_xp}^{+}\right)^\top$, we state the following TLS problem:
	\begin{subequations} \label{E27}
		\begin{align}
			\min_{\tilde{\mathcal{H}}_{n_xp}} \ &\norm{\tilde{\mathcal{H}}_{n_xp}^\top}_F, \label{E27a}\\
			\text{s.t.} \ &\left(\hat {\mathcal{H}}_{n_xp}^{-} + \tilde  {\mathcal{H}}_{n_xp}^{-}\right)^\top \in \mathcal{R}\left((\hat {\mathcal{H}}_{n_xp}^{+} + \tilde  {\mathcal{H}}_{n_xp}^{+})^\top\right), \label{E27b}
		\end{align}
	\end{subequations}
	where $\tilde {\mathcal{H}}_{n_xp} = \begin{bmatrix}\tilde {\mathcal{H}}_{n_xp}^{+}\\ \tilde  {\mathcal{H}}_{n_xp}^{-}\end{bmatrix}$.
	To solve this problem, similar to the SVD \eqref{E7} of ${\mathcal{H}}_{n_xp}$, let
	\begin{equation} \label{E28}
		\hat {\mathcal{H}}_{n_xp} = \begin{bmatrix}\hat {\mathcal{H}}_{n_xp}^{+}\\\hat {\mathcal{H}}_{n_xp}^{-}\end{bmatrix} = \hat U\hat S {\hat V}^\top,
	\end{equation}
	where
	\begin{equation*}
		\begin{split}
			\hat U & = \begin{bmatrix}\hat u_1&\hat u_2&\cdots&\hat u_{n_x+1}\end{bmatrix}, \\
			\hat V & = \begin{bmatrix}\hat v_1&\hat v_2&\cdots&\hat v_{p}\end{bmatrix}, \\
			\hat S & = \begin{bmatrix}\text{diag}\left(\hat \sigma_1,\hat \sigma_2,\dots,\hat \sigma_{n_x+1}\right)&0\end{bmatrix}.
		\end{split}
	\end{equation*}
	Assuming that $\hat \sigma_{n_x} > \hat \sigma_{n_x+1}$, the unique solution to the above TLS problem is
	\begin{equation} \label{E29}
		\hat {\bm{a}}_\text{tls}^\top = \hat u_{n_x+1,1:n_x}\hat u_{n_x+1,n_x+1}^{-1},
	\end{equation}
	or equivalently,
	\begin{equation} \label{E30}
		\hat{\bm{a}}_\text{tls}^\top = -\left(\hat {\mathcal{H}}_{n_xp}^{+}(\hat {\mathcal{H}}_{n_xp}^{+})^{\top}-\hat \sigma_{n_x+1}^2I\right)^{-1}\hat {\mathcal{H}}_{n_xp}^{+}\left(\hat {\mathcal{H}}_{n_xp}^{-}\right)^\top.
	\end{equation}
	
	\begin{remark} \label{Rmk2}
		Corresponding to the form \eqref{E30a} for OLS, an equivalent formulation for the TLS problem \eqref{E27} is given in \cite{Golub2013matrix}, which states that the TLS solution $\hat{\bm{a}}_\text{tls}$ is the minimum norm solution of
		\begin{equation} \label{E30b}
			\bm{a}{\doublehat {\mathcal{H}}}_{n_xp}^{+} = \doublehat {\mathcal{H}}_{n_xp}^{-},
		\end{equation}
		where $\begin{bmatrix}\doublehat {\mathcal{H}}_{n_xp}^{+}\\\doublehat {\mathcal{H}}_{n_xp}^{-}\end{bmatrix} = \sum_{i=i}^{n_x}\hat \sigma_{i}\hat u_i{\hat v}_i^\top$ is the TLS approximation of $\hat {\mathcal{H}}_{n_xp}$. As shown in the next subsection, the forms \eqref{E30a} and \eqref{E30b} are useful for the sensitivity analysis of least-squares.
	\end{remark}

	After replacing the coefficients $\left\{a_i\right\}_{i=1}^{n_x}$ in $\bm{a}$ with estimates in $\hat{\bm{a}}_\text{tls}$, we obtain an estimate of the $A$-matrix on the observer canonical form \eqref{E16}, denoted by ${\hat A}_{\text{tls}}$. At this point, we have formulated a TLS problem, with $\bm{a}$ as decision variable, and solved it using SVD, resulting in an $A$-matrix on observer canonical form. However, this approach has not yet been linked to the RASBR algorithm \eqref{E9} which estimates the $A$-matrix via the SVD \eqref{E28} as follows:
	\begin{equation}\label{E31}
		{\hat A}_{R} = \left((\hat {\mathcal{O}}_{n_x}^{+})^{\top}\hat {\mathcal{O}}_{n_x}^{+}\right)^{-1}(\hat {\mathcal{O}}_{n_x}^{+})^{\top}\hat {\mathcal{O}}_{n_x}^{-},
	\end{equation}
	where $\hat {\mathcal{O}}_{n_x}^{-}$ and $\hat {\mathcal{O}}_{n_x}^{+}$ are the last and first $n_x$ rows of $\hat{\mathcal{O}}_{n_x}$, respectively, and $\hat {\mathcal{O}}_{n_x} = \hat U_1\hat S_1^{1/2}$ with $\hat S_1 = \text{diag}\left(\hat \sigma_1,\hat \sigma_2,\dots,\hat \sigma_{n_x}\right)$, and $\hat U_1$ containing left singular vectors corresponding to $\hat S_1$.
	
	The following theorem shows that the TLS solution \eqref{E30} is equivalent to the RASBR's solution \eqref{E31}, in the sense that they return the same matrix $A$, up to a similarity transformation.
	\begin{theorem} \label{Thm0}
		Matrices ${\hat A}_{\text{tls}}$ and ${\hat A}_{R}$, defined by \eqref{E30} and \eqref{E31}, respectively, are similar, i.e., 
		there exists an invertible matrix $T$, such that ${\hat A}_{\text{tls}} = T^{-1}{\hat A}_{R}T$.
	\end{theorem}
	\begin{proof}
		See Appendix \ref{App0}.
	\end{proof}
	
	Since matrices ${\hat A}_{R}$ and ${\hat A}_{\text{tls}}$ are similar, representing the same linear mapping under two possibly different bases, we conclude that RASBR is a TLS solution. Furthermore, Theorem \ref{Thm0} indicates that RASBR admits a solution ${\hat A}_{\text{tls}}$ in the null space of the extended observability matrix $\hat {\mathcal{O}}_{n_x}$. In other words, the RASBR algorithm simultaneously seeks to recover the null space of $\hat {\mathcal{O}}_{n_x}$, although its solution ${\hat A}_{R}$ is explicitly based on the range space of $\hat {\mathcal{O}}_{n_x}$. This provides a unified framework for comparing RASBR and NUSBR, enabling a performance evaluation simply by comparing the solutions $\hat{\bm{a}}_\text{ols}$ \eqref{E26} and $\hat{\bm{a}}_\text{tls}$ \eqref{E30}.
	
	\begin{remark} \label{Rmk3}
		In our presentation, the future horizon is fixed at $f = n_x$ for analysis and comparison. Although it is not necessary to make this choice, for instance, a typical choice for many SIMs is to take $f = p$, our perspective that NUSBR is an OLS solution and RASBR is a TLS solution remains for other choices of $f$. To explain this, for the case $f > n_x$, we have $\text{dim}\left(\mathcal{K}({\mathcal{O}}_{f}^\top)\right) = f+1-n_x$. Based on the shift-invariant structure of $\mathcal{O}_{f}$ and the Cayley-Hamilton theorem, the null-space of $\mathcal{O}_{f}$ can be fully parameterized in terms of the rows of the matrix
		\begin{equation} \label{E33a}
			N_f(\bm{a}) = \begin{bmatrix}
				a_{n_x}& \cdots &a_1&1&0&\cdots&0 \\
				0&a_{n_x}& \cdots &a_1&1&\cdots&0 \\
				\vdots& \ddots& \ddots&  \ddots& \ddots& \ddots& \vdots\\
				0&\cdots&0&a_{n_x}&a_{n_x-1}& \cdots &1 \\
			\end{bmatrix},
		\end{equation}
		where $N_f(\bm{a}) \in \mathbb{R}^{(f+1-n_x)\times (f+1)}$ is a banded Toeplitz matrix having $\begin{bmatrix}\bm{a}&1&0&\cdots&0\end{bmatrix}$ as its first row and $\begin{bmatrix}{a}_{n_x}&0&\cdots&0\end{bmatrix}^\top$ as its first column  \cite{De2019least}. Leveraging the shift-invariant structure of  $N_f(\bm{a})$ and ${\mathcal{H}}_{fp}$, it can be shown that the condition $N_f(\bm{a}){\mathcal{O}}_{f} = N_f(\bm{a}){\mathcal{H}}_{fp} = 0$ is equivalent to imposing $\begin{bmatrix}\bm{a}&1\end{bmatrix}{\mathcal{H}}_{n_xp} = 0$, which corresponds to the case $f = n_x$. Therefore, there is no compelling reason to select $f > n_x$ for the NUSBR method, as doing so results in the same OLS solution.
		
		For the RASBR, matrix $A$ is similarly given by \eqref{E31}, 
		\begin{equation}\label{E33b}
			{\hat A}_{R} = \left((\hat {\mathcal{O}}_{f}^{+})^{\top}\hat {\mathcal{O}}_{f}^{+}\right)^{-1}(\hat {\mathcal{O}}_{f}^{+})^{\top}\hat {\mathcal{O}}_{f}^{-}.
		\end{equation}
		Meanwhile, after partitioning matrix $\hat U$ from the SVD \eqref{E28} as
		\begin{equation}\label{E33c}
			\hat U = \begin{bmatrix}\hat u_1&\hat u_2&\cdots&\hat u_{f+1}\end{bmatrix}=\begin{bmatrix}\hat U_{11}&\hat U_{12}\\\hat U_{21}&\hat U_{22}\end{bmatrix},
		\end{equation}
		where $\hat U_{11} \in \mathbb{R}^{n_x\times n_x}$ and $\hat U_{22} \in \mathbb{R}^{({f+1-n_x})\times ({f+1-n_x})}$, we have a unique solution of the TLS problem \eqref{E27}, which is
		\begin{equation} \label{E33d}
			\hat {\bm{a}}_\text{tls}^\top = \hat U_{12}{\hat U}_{22}^{-1} \in \mathbb{R}^{n_x\times (f+1-n_x)}.
		\end{equation}
		Similar to Theorem \ref{Thm0}, the rows of the matrix $\begin{bmatrix}\hat {\bm{a}}_\text{tls}&I\end{bmatrix}$ span the left null space of $\hat {\mathcal{O}}_{f}$. After normalizing $\begin{bmatrix}\hat {\bm{a}}_\text{tls}&I\end{bmatrix}$ to have the same banded Toeplitz structure as matrix $N_f(\bm{a})$ in \eqref{E33a}, we can bridge RASBR and TLS. Moreover, in our simulations, we found that for finite sample settings, taking $f = n_x$ gives similar performance to $f > n_x$ for the RASBR method. For further discussions about the impact of $f$, see \cite{Qin2006role,Chiuso2007some}.
	\end{remark}
	
	\subsection{Sensitivity Analysis} \label{Sct3.3}
	
	In this subsection we establish some inequalities that shed light on the sensitivity of least-squares problems as well as on the relationship between $\hat{\bm{a}}_\text{tls}$ and $\hat{\bm{a}}_\text{ols}$.
	First, from \eqref{E26} and \eqref{E30} we have
	\begin{equation}  \label{E34}
		\hat{\bm{a}}_\text{tls} -  \hat{\bm{a}}_\text{ols}=-\hat\sigma_{n_x+1}^2\hat{\bm{a}}_\text{ols}\left(\hat {\mathcal{H}}_{n_xp}^{+}(\hat {\mathcal{H}}_{n_xp}^{+})^{\top}-\hat \sigma_{n_x+1}^2I\right)^{-1},	
	\end{equation}
	which suggests that when $\hat\sigma_{n_x+1} = 0$, $\hat{\bm{a}}_\text{tls} = \hat{\bm{a}}_\text{ols}$. This corresponds to the case when the Markov parameters are exact and the Hankel matrix is low rank. Besides this case, the solutions of $\hat{\bm{a}}_\text{tls}$ and $\hat{\bm{a}}_\text{ols}$ are generally different.
	
	To effectively compare the performance of these two solutions, a suitable metric is needed. We first introduce the following definition which quantifies the degree of closeness between a vector and a subspace:
	
	\begin{definition}[$\sin\theta$ theory \cite{Wedin1972perturbation}] \label{Def1}
		The angle $\theta\left(\alpha,S\right)$ between a vector $\alpha$ with $\norm{\alpha}= 1$ and a subspace $S$ is defined by
		\begin{equation} \label{E34b}
			\sin\theta\left(\alpha,S\right) := \min_{\beta\in S}\norm{\alpha-\beta} = \norm{\left(I-P_S\right)\alpha},
		\end{equation}
		where $P_S$ is the orthogonal projection onto $S$.
	\end{definition}
	
	The above definition is extended to the following definition to quantify the distance between two subspaces:
	\begin{definition}[{\cite[Section 2.5.3]{Golub2013matrix}}] \label{Def2}
		The distance between two subspaces $S_1$ and $S_2$ is defined by
		\begin{equation} \label{E34c}
			\norm{\sin\theta\left(S_1,S_2\right)} := \norm{P_1-P_2},
		\end{equation}
		where $P_i$ is the orthogonal projection onto $S_i$ for $i=1,2$.
	\end{definition}
	
	Based on the above two definitions, we have the following lemma which bounds the distances of $\hat{\bm{a}}_\text{tls}$ and $\hat{\bm{a}}_\text{ols}$ to the true value $\bm{a}$.
	\begin{lemma} \label{Lem0}
		The distances of the solutions of TLS and OLS to the true value $\bm{a}$ satisfy
		\begin{subequations} \label{E34d}
			\begin{align}
				\nonumber
				&\frac{\norm{\hat{\bm{a}}_\text{tls} -  \bm{a}}}{\norm{\begin{bmatrix}\bm{a}&1
				\end{bmatrix}}} = \norm{\sin\theta\left(\mathcal{K}\left(\doublehat{\mathcal{H}}_{n_xp}^\top\right),\mathcal{K}\left({\mathcal{H}}_{n_xp}^\top\right)\right)}  \\			
				&= \norm{\sin\theta\left(\mathcal{R}\left(\doublehat{\mathcal{H}}_{n_xp}\right),\mathcal{R}\left({\mathcal{H}}_{n_xp}\right)\right)} \leq \frac{\norm{\tilde{\mathcal{H}}_{n_xp}}}{\hat\sigma_{n_x}}, \label{E34d1}\\
				\nonumber
				&\frac{\norm{\hat{\bm{a}}_\text{ols} -  \bm{a}}}{\norm{\begin{bmatrix}\bm{a}&1
				\end{bmatrix}}} = \norm{\sin\theta\left(\mathcal{K}\left({\hat{\bar{\mathcal{H}}}_{n_xp}^\top}\right),\mathcal{K}\left({\mathcal{H}}_{n_xp}^\top\right)\right)} \\
				& = \norm{\sin\theta\left(\mathcal{R}\left(\hat{\bar{\mathcal{H}}}_{n_xp}\right),\mathcal{R}\left({\mathcal{H}}_{n_xp}\right)\right)} \leq \frac{\norm{\tilde{\mathcal{H}}_{n_xp}^{+}}}{{\hat\sigma}_{n_x}^{+}} \label{E34d2},
			\end{align}
		\end{subequations}
		where $\doublehat {\mathcal{H}}_{n_xp} = \begin{bmatrix}\doublehat {\mathcal{H}}_{n_xp}^{+}\\\doublehat {\mathcal{H}}_{n_xp}^{-}\end{bmatrix}$ and $\hat{\bar{\mathcal{H}}}_{n_xp} = \begin{bmatrix}
			\hat{{\mathcal{H}}}_{n_xp}^{+}\\\hat{\bar{\mathcal{H}}}_{n_xp}^{-}
		\end{bmatrix}$ are given in \eqref{E30b} and \eqref{E30a}, and $\hat\sigma_{n_x}$ with ${\hat\sigma}_{n_x}^{+}$ are the $n_x$-th largest singular values of $\hat{{\mathcal{H}}}_{n_xp}$ and  $\hat{{\mathcal{H}}}_{n_xp}^{+}$, respectively.
	\end{lemma}
	\begin{proof}
		See Appendix \ref{App1}.
	\end{proof}
	
	\textbf{When TLS is better than OLS}: The terms $\norm{\tilde{\mathcal{H}}_{n_xp}}$ and $\norm{\tilde{\mathcal{H}}_{n_xp}^{+}}$ are related to the inaccuracy of the Hankel matrix $\hat{\mathcal{H}}_{n_xp}$. Since $\tilde {\mathcal{H}}_{n_xp} = \begin{bmatrix}\tilde {\mathcal{H}}_{n_xp}^{+}\\ \tilde  {\mathcal{H}}_{n_xp}^{-}\end{bmatrix}$, where $\tilde{\mathcal{H}}_{n_xp}^{+}$ contains $n_x$ rows and $\tilde  {\mathcal{H}}_{n_xp}^{-}$ contains the last row, we can approximate the norm as $\norm{\tilde{\mathcal{H}}_{n_xp}} \approx\norm{\tilde{\mathcal{H}}_{n_xp}^{+}}$. Then, Lemma~\ref{Lem0} suggests that the difference in sensitivity of TLS and OLS depends on the ratio
	\begin{equation}
		\frac{\norm{\hat{\bm{a}}_\text{ols} -  \bm{a}}}{\norm{\hat{\bm{a}}_\text{tls} -  \bm{a}}} \approx \frac{\hat\sigma_{n_x}}{\hat\sigma_{n_x}^{+}} =: \hat\kappa.
	\end{equation}
	Based on interlacing inequalities for singular values \cite{Golub2013matrix}, we have that $\hat\sigma_{n_x} \geq \hat\sigma_{n_x}^{+}$, which implies that $\hat\kappa \geq 1$. Therefore, the above ratio $\hat\kappa$ is not informative about when OLS is better than TLS. However, it suggests that the better accuracy of TLS as opposed to OLS will be more pronounced when the ratio $\hat\kappa$ is larger. This statement is supported by the following simulation result. 
	
	\textbf{Experiment 2}: We use the following MATLAB code to generate random systems:
	\begin{equation*}
		\begin{split}
			&{\rm{m = idss(drss(n_x,1,1));}}\\
			&{\rm{m.d = zeros(1,1);}}\\
			&{\rm{m.b = 5*randn(n_x,1);}}
		\end{split}
	\end{equation*}
	Following the suggestion in \cite{Rojas2015critical}, the magnitude of the sampled dominant pole $p_{\rm{max}}$ is restricted to satisfy $0.78<p_{\rm{max}}<0.9$, and the system dimension $n_x=2$ and number of Markov parameters $n=20$. Next, we evaluate the performance of TLS and OLS using the corrupted Markov parameters of these random systems, where independent Gaussian noise with the variance of 0.5 is added to the true Markov parameters. Based on the ratio $\hat\kappa$ of these random systems, we divide these random systems into three groups, where each group consists of 200 Monte Carlo trials (corresponding to 200 random systems). The parameter settings for these groups are summarized in Table~\ref{Table1a}. The performance is evaluated by the $\text{FIT}$ defined in \eqref{E17a}, and the results are shown in Figure~\ref{F3a}. As shown from the figure, from Groups A to C, the better accuracy of TLS compared to OLS becomes increasingly evident as the ratio $\hat\kappa$ increases. 
	
	\begin{table}[H]
		\caption {Parameter Settings of Experiment 2} \label{Table1a}
		\begin{center}
			\begin{tabular}{cccccc}
				\toprule
				{Group} & spectral radius of $A$ & $n_x$ & $n$ & $\hat\kappa = {\hat\sigma}_{n_x}/{\hat\sigma}_{n_x}^{+}$ \\
				\midrule
				A & $0.78<\rho(A)<0.9$ &  2  &  20 & $1.0<\hat\kappa<1.1$   \\
				B & $0.78<\rho(A)<0.9$ &  2  &  20 & $1.3<\hat\kappa<1.4$  \\
				C & $0.78<\rho(A)<0.9$ &  2  &  20 & $1.6<\hat\kappa<1.7$  \\
				\bottomrule
			\end{tabular}
		\end{center}
	\end{table}

	\begin{figure}
		\centering
		\includegraphics[scale=0.6]{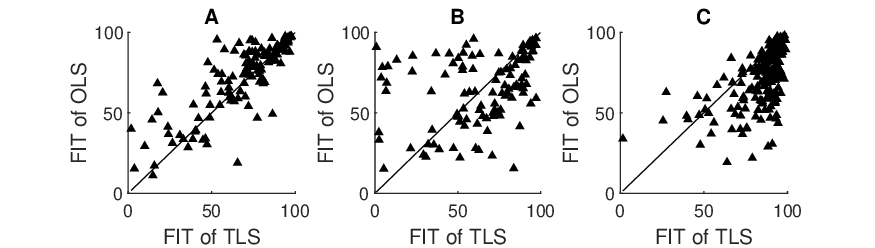}
		\caption{FITs of impulse responses for three experiment groups in Table~\ref{Table1a}: A random system ($\blacktriangle$), and the solid line is a bisector line.}
		\label{F3a}
	\end{figure}
	
	\begin{remark}
		The factors which influence the ratio $\hat\kappa$, or equivalently $\kappa := \sigma_i / \sigma_i^{+}$: It was suggested in \cite{Van1989accuracy} that the ratio $\hat\kappa$ is mainly influenced by the length and the orientation of $\mathcal{H}_{n_xp}^{-}$ w.r.t. the singular vectors of $\mathcal{H}_{n_xp}^{+}$. Here we briefly summarize their impact. Before proceeding, we mention that, similar to the SVD of $\hat{\mathcal{H}}_{n_xp}^{+}$ and $\hat{\mathcal{H}}_{n_xp}$ in \eqref{E26a} and \eqref{E28}, the corresponding SVD of the true matrices ${\mathcal{H}}_{n_xp}^{+}$ and ${\mathcal{H}}_{n_xp}$ is obtained by omitting the hat notation in \eqref{E26a} and \eqref{E28}.	
		
		(1) The orientation of the observation matrix $\mathcal{H}_{n_xp}^{-}$ w.r.t. the lowest singular vectors $v_i^{+}$ of $\mathcal{H}_{n_xp}^{+}$ plays an important role. The extent to which $\sigma_i$ is larger depends on the orientation of $\mathcal{H}_{n_xp}^{-}$ w.r.t. $v_i^{+}$ of $\mathcal{H}_{n_xp}^{+}$. It is not difficult to see that those directions $v_i^{+}$ of $\mathcal{H}_{n_xp}^{+}$ onto which $\mathcal{H}_{n_xp}^{-}$ has a larger projection will be favored and have the largest ratio $\kappa$. Hence, the ratio $\kappa$ increases when $\mathcal{H}_{n_xp}^{-}$ becomes closer and closer to the direction of $v_{n_x}^{+}$. This happens not only if $\mathcal{H}_{n_xp}^{-}$ is close to the lowest singular vector $u_{n_x}^{+}$, but also if the condition number $\sigma_1^{+} / \sigma_{n_x}^{+}$ of $\mathcal{H}_{n_xp}^{+}$ is large. Indeed, a large condition number implies a small $\sigma_{n_x}^{+}$. Since the matrix $\mathcal{H}_{n_xp}$ is low rank, we have $\left(\mathcal{H}_{n_xp}^{-}\right)^\top = \sum_{i=1}^{n_x} \beta_i v_i^{+}$, thus the exact solution is $\bm{a}^\top = \sum_{i=1}^{n_x} \frac{\beta_i}{\sigma_i^{+}} v_i^{+}$. Hence, the component $\beta_{n_x} / \sigma_{n_x}^{+}$ of $\bm{a}$ in the direction of $v_{n_x}^{+}$ is large whenever $\sigma_{n_x}^{+}$ is small. It is important to note that $\sigma_{n_x}^{+}$ should not be too small, as this may result in an ill-conditioned problem, leading to different conclusions. This issue is discussed later. The better accuracy of TLS in comparison to OLS is most pronounced when each column of $\mathcal{H}_{n_xp}^{-}$ is parallel to $u_{n_x}^{+}$. In this case, the ratio $\kappa$ reaches its maximal value for a given $\norm{\mathcal{H}_{n_xp}^{-}}$ and $\sigma_n^{+}$ is given by 
		\begin{equation}  \label{EA36}
			\sigma_n = \sqrt{(\sigma_{n_x}^{+})^2 + \norm{\mathcal{H}_{n_xp}^{-}}^2}.
		\end{equation}
		
		(2) The length of each observation vector in $\mathcal{H}_{n_xp}^{-}$ influences the ratio $\kappa$. This is clear from \eqref{EA36}. Making $\norm{\mathcal{H}_{n_xp}^{-}}$ larger increases the ratio $\kappa$ if $\mathcal{H}_{n_xp}^{-}$ is parallel to $u_{n_x}^{+}$ and hence, the better accuracy of the TLS solution w.r.t. the OLS solution will be more pronounced.
	\end{remark}
	
	\textbf{Experiment 1} (Revisited): Returning to Systems 1 and 2 in Section \ref{Sct1}, we compute the ratio $\kappa = {\sigma}_{n_x}/{\sigma}_{n_x}^{+}$ for the true systems and average $\hat\kappa = \hat{\sigma}_{n_x}/\hat{\sigma}_{n_x}^{+}$ for the 200 random noisy cases. We find that for System 1, $\kappa_1 = 1.0104$ and $\hat\kappa_1 = 1.0320$, and for System 2, $\kappa_2 = 1.7890$ and $\hat\kappa_2 = 1.7286$. Based on our argument that the better accuracy of TLS w.r.t. OLS will be more pronounced when the ratio $\hat\kappa$ is larger, this result explains why TLS outperforms OLS for System~2. However, this is not sufficient to explain why OLS performs better than TLS for System~1. Since the ratio $\hat\kappa$ becomes larger when $\hat\sigma_{n_x}^{+}$ becomes smaller, it might be expected that TLS would consistently outperform OLS when $\hat\sigma_{n_x}^{+}$ is very small. However, this is not the case due to the ill-conditioning of the least-squares problem. In such situations, as we will see later, the conditioning of TLS is always worse than that of OLS, leading OLS to provide a more reliable solution. 
	
	\textbf{When OLS is better than TLS}: We now discuss the ill-conditioning of least-squares, which explains why OLS performs better than TLS for System 1 in Experiment 1. It was shown in \cite{Golub1980analysis} that the difference between OLS and TLS grows with
	\begin{equation} \label{E35A}
		\hat\delta:={\hat\sigma}_{n_x}^{+} - \hat\sigma_{n_x+1}.
	\end{equation}
	This is summarized in the following lemma:
	
	\begin{lemma} [{\cite[Corollary 4.2]{Golub1980analysis}}] \label{Lem1}
		If ${\hat\sigma}_{n_x}^{+} > \hat\sigma_{n_x+1}$, then
		\begin{subequations}  \label{E35}
			\begin{align}
				&\norm{\hat{\bm{a}}_\text{tls} - \hat{\bm{a}}_\text{ols}} \leq \frac{\norm{{\hat {\mathcal{H}}}_{n_xp}^{-}}\rho_{\text{ols}}}{\left({\hat\sigma}_{n_x}^{+}\right)^2 - \hat\sigma_{n_x+1}^2}, \label{E35a}\\
				&\rho_{\text{tls}} \leq \rho_{\text{ols}}\left(1+\frac{\norm{{\hat {\mathcal{H}}}_{n_xp}^{-}}}{\hat\delta}\right), \label{E35b}
			\end{align}
		\end{subequations}
		where $\rho_{\text{ols}} =\norm{{\hat {\mathcal{H}}}_{n_xp}^{-} - \hat{\bm{a}}_\text{ols}\hat {\mathcal{H}}_{n_xp}^{+}}$, $\rho_{\text{tls}} =\norm{{\hat {\mathcal{H}}}_{n_xp}^{-} - \hat{\bm{a}}_\text{tls}\hat {\mathcal{H}}_{n_xp}^{+}}$.
	\end{lemma}
	
	According to \eqref{E35a} in Lemma \ref{Lem1}, the upper bound for $\norm{\hat{\bm{a}}_\text{tls} - \hat{\bm{a}}_\text{ols}}$ becomes larger when $\hat\sigma_{n_x+1}$ is close to ${\hat\sigma}_{n_x}^{+}$. Moreover, \eqref{E35b} suggests that the gap between the costs of TLS and OLS is larger whenever $\hat\sigma_{n_x+1}$ is close to ${\hat\sigma}_{n_x}^{+}$. The following lemma indicates the extent to which the inverse of the difference $\hat\delta$ defined in \eqref{E35A} measures the sensitivity of the TLS problem:
	\begin{lemma} [{\cite[Lemma 4.3]{Golub1980analysis}}] \label{Lem2}
		If ${\hat\sigma}_{n_x}^{+} > \hat\sigma_{n_x+1}$, then
		\begin{equation}  \label{E36}
			\frac{\norm{{\hat {\mathcal{H}}}_{n_xp}^{-}{\hat v}_{n_x}}}{2\hat\delta} \leq \norm{\hat{\bm{a}}_\text{tls}} \leq  \frac{\norm{{\hat {\mathcal{H}}}_{n_xp}^{-}}}{\hat\delta}.
		\end{equation}
	\end{lemma}
	
	Lemma \ref{Lem2} implies that the TLS solution is unstable whenever $\hat\sigma_{n+1}$ is close to ${\hat\sigma}_{n_x}^{+}$. In other words, $\hat\delta$ is a measure of how close the TLS problem \eqref{E21} is to the class of insoluble TLS problems. Moreover, it was claimed in \cite{Golub1980analysis} that TLS is a deregularization of  OLS. In this sense, the conditioning of TLS problems is always worse than that of OLS. We now establish bounds for $\hat\delta$, which provide insight into when the TLS solution becomes unstable and the OLS might yield a better solution than TLS.
	\begin{lemma} \label{Lem3}
		The distance $\hat\delta$ between $ {\hat\sigma}_{n_x}^{+}$ and $\hat\sigma_{n_x+1}$ satisfies
		\begin{equation}  \label{E36a}
			\hat\delta \leq {\hat\sigma}_{n_x}^{+} \leq \Phi(A)^2\norm{C}\norm{B}\frac{1-\rho(A)^{n-n_x}}{{1-\rho(A)}} + \norm{\tilde{\mathcal{H}}_{n_xp}^{+}}, 	
		\end{equation}
		where $\Phi(A) := \sup_{\tau \geq 0}\frac{\norm{A^\tau}}{\rho(A)^{\tau/2}}$, which is finite.
	\end{lemma}
	\begin{proof}
		See Appendix \ref{App2}.
	\end{proof}
	
	The upper bound on the right side of \eqref{E36a} sheds light on when $\hat\sigma_{n+1}$ is close to ${\hat\sigma}_{n_x}^{+}$. To be specific, if ${\hat\sigma}_{n_x}^{+}$ is small, or equivalently, matrix ${\hat {\mathcal{H}}}_{n_xp}^{+}$ tends to be rank-deficient, then $\hat\sigma_{n_x+1}$ is likewise close to ${\hat\sigma}_{n_x}^{+}$. Moreover, the term $\frac{1-\rho(A)^{(n-n_x)}}{{1-\rho(A)}}$ in the upper bound for ${\hat\sigma}_{n_x}^{+}$ indicates that ${\hat\sigma}_{n_x}^{+}$ becomes small when either the spectral radius $\rho(A)$ is close to zero or the dimension $n_x$ is large. These statements are supported by simulation results.
	
	\textbf{Experiment 3}: Same as the Experiment 2, we use the MATLAB command $\text{drss}(\cdot)$ to generate random systems. Next, we evaluate the performance of TLS and OLS using the corrupted Markov parameters of these random systems, where independent Gaussian noise with the variance of 0.5 is added to the true Markov parameters. We conduct six groups of experiments, and run 200 Monte Carlo trials in each group. The parameter settings for these groups are summarized in Table~\ref{Table1}. These six experiments fall into two families: Family A varies the spectral radius of $A$, keeping the system dimension $n_x$ and the number of Markov parameters $n$ fixed. Family B varies the system dimension $n_x$, while the spectral radius of $A$ and the number of Markov parameters $n$ are fixed. Moreover, the average $\hat\delta$ of these randoms systems in each group is given in Table~\ref{Table1}.
	\begin{table}[H]
		\caption {Parameter Settings of Experiment 3} \label{Table1}
		\begin{center}
			\begin{tabular}{cccccc}
				\toprule
				{Groups} & spectral radius of $A$ & $n_x$ & $n$ &$\text{mean}\left(\hat\delta\right)$\\
				\midrule
				A1 & $0.85<\rho(A)<0.95$ &  2  &  20 &  2.9027    \\
				A2 & $0.55<\rho(A)<0.65$ &  2  &  20 &  0.5536   \\
				A3 & $0.05<\rho(A)<0.15$ &  2  &  20 &  0.2635   \\
				B1 & $0.78<\rho(A)<0.9$ &  2 &   50 &  1.6648   \\
				B2 & $0.78<\rho(A)<0.9$ &  6 &   50 &  0.1995   \\
				B3 & $0.78<\rho(A)<0.9$ &  10 &  50 &  0.1782   \\
				\bottomrule
			\end{tabular}
		\end{center}
	\end{table}
	
	\begin{figure}
		\centering
		\includegraphics[scale=0.6]{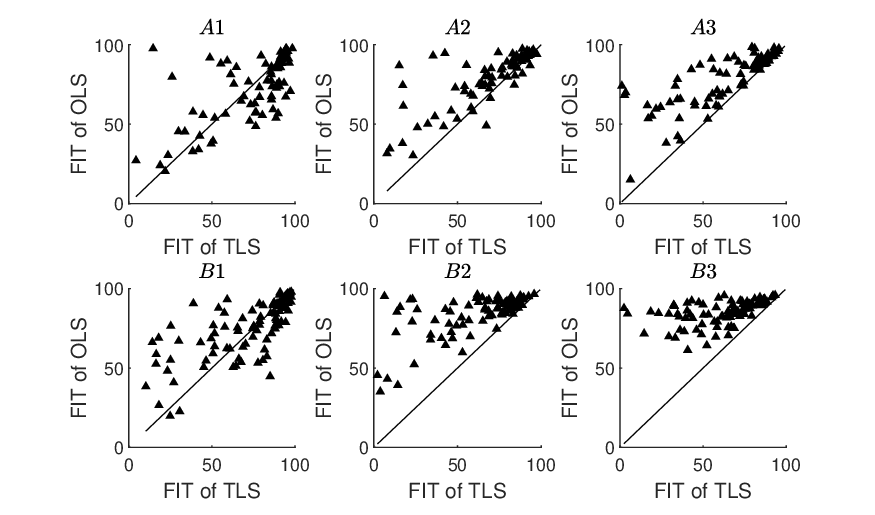}
		\caption{FITs of impulse responses for six experiment groups in Table \ref{Table1}: A random system ($\blacktriangle$), and the solid line is a bisector line.}
		\label{F3}
	\end{figure}
	
	The performance is evaluated by the $\text{FIT}$ defined in \eqref{E17a}, and the results are shown in Figure~\ref{F3}. As shown from the figure, for Experiments $A1$, $A2$ and $A3$ in family $A$, when the spectrum radius of $A$ approaches zero, OLS outperforms TLS. For Experiments $B1$, $B2$ and $B3$ in family $B$, as the system order $n_x$ increases, OLS performs better than TLS. Meanwhile, as shown in the last column of Table~\ref{Table1}, the difference $\hat\delta$ becomes smaller when the spectrum radius of $A$ approaches zero, or when the system dimension increases. These results verify that for the summarized scenarios where the TLS problem is unstable, OLS may give a better performance than TLS. However, it should be mentioned that the bound on the right side of \eqref{E36a} is just an upper bound. To reflect whether additional factors have contributed to the ill-conditioning of the realization problem, a precise quantification of the difference ${\hat\sigma}_{n_x}^{+} - \hat\sigma_{n_x+1}$ is further required.
	
	\textbf{Experiment 1} (Revisited): Let us return to the Systems 1 and 2 in the Introduction, where we notice that $\rho(A)$ of System 1 is close to zero. We now compute the difference $\delta := {\sigma}_{n_x}^{+}-{\sigma}_{n_x+1}$ for the true systems and average $\hat\delta={\hat\sigma}_{n_x}^{+}-{\hat\sigma}_{n_x+1}$ for the 200 random noisy cases. We find that for System 1, $\delta_1 = 1.8463$ and the average difference is $\hat\delta_1 = 0.6742$, and for System 2, $\delta_2 = 19.8737$ and the average difference is $\hat\delta_2 = 7.8744$. Therefore, compared with System 2, the realization problem for System 1 is likewise ill-conditioned. This is the reason why OLS is better than TLS on System 1.
	
	In summary, for a well-conditioned realization problem, the better accuracy of TLS as opposed to OLS will be more pronounced when the ratio $\hat\kappa$ is larger. Conversely, for a likewise ill-conditioned problem, since the conditioning of TLS is always worse than OLS, OLS may perform better than TLS. 
	
	\section{Optimal Realization} \label{Sct4}
	
	As shown in Sections \ref{Sct1} and \ref{Sct3}, the performance of OLS and TLS is case-dependent, which suggests that none of them is optimal. The primary reason lies in the fact that the assumptions underlying both OLS and TLS are not fully satisfied in the context of the realization problem. For instance, the upper part of the Hankel matrix, $\hat{\mathcal{H}}_{n_xp}^{+}$, is also affected by noise, rendering the OLS assumption of an accurate data matrix unrealistic. Similarly, due to the inherent Hankel structure, the rows of $\hat{\mathcal{H}}_{n_xp}^{+}$ are not independent, and they do not have the same covariance equal to $\hat{\sigma}_{n_x+1}^2I$. This means that the correction term $\hat{\sigma}_{n_x+1}^2I$ used in the TLS solution \eqref{E30} is ineffective in practice. Therefore, while the Hankel structure is convenient for the realization purpose, it also introduces challenges in managing errors due to noise. Notably, both OLS and TLS fail to utilize the Hankel matrix structure when mitigating noise effects, making them suboptimal and highly case-dependent.
	
	As well recognized in the literature of SIMs \cite{Qin2006overview}, the statistical and numerical properties of realization algorithms can be improved by pre- and post-multiplying the Hankel matrix $\hat {\mathcal{H}}_{fp}$ with weighting
	matrices before the SVD. In range-space-based methods, several candidate weighting matrices have been proposed \cite{Qin2006overview}; however, the optimal weighting that minimizes the variance remains undetermined. As highlighted in \cite{Viberg1997analysis} and recent work \cite{Galrinho2014weighted}, the null space fitting method provides a possibility to derive an optimal weighting by effectively leveraging the structure of the Hankel matrix. Building on these findings, we propose a WLS solution for the estimation of state-space models\footnote{A preliminary version of this WLS realization method appeared in \cite{He2024weighted}.}. Simulation results and statistical analysis demonstrate that the proposed approach achieves minimal variance, outperforming the OLS and TLS solutions in this respect.
		
	\subsection{Preliminaries on Weighted Least-Squares} \label{Sct4.0}	
	Before we proceed, we give a short introduction to WLS. Using the same notations as in Section~\ref{Sct2}, we consider the linear regression model
    \begin{equation}\label{eq:wls_model}
    	z = \Phi\theta +\tilde z,
    \end{equation}
    where the noise satisfies $\mathbb{E}\left\{\tilde z\right\} = 0$ and $\mathbb{E}\left\{\tilde z\,\tilde z^\top\right\} = P_{z} \succ 0$. Defining the weighting matrix
    \begin{equation}
        W = P_{z}^{-1},
    \end{equation}
    the WLS estimator is then given by
    \begin{equation} \label{eq:wls_estimator}
    	\hat\theta_{\text{wls}}
    	=\bigl(\Phi^\top W\Phi\bigr)^{-1}\Phi^\top Wz.
    \end{equation}
    Equivalently, $\hat\theta_{\text{wls}}$ can be characterized as the minimizer of the weighted residual sum of squares:
    \begin{equation} \label{eq:wls_optimization}
    	\hat\theta_{\text{wls}}=
    	\arg\min_{\theta}
    	\bigl(z - \Phi\,\theta\bigr)^\top
    	W\bigl(z - \Phi\,\theta\bigr).
    \end{equation}
    
    When the noise is homoskedastic, i.e., $P_{z} = \sigma^{2}I$, the WLS estimator~\eqref{eq:wls_estimator} reduces exactly to the OLS estimator~\eqref{E20}. In the general heteroskedastic case, the WLS estimator is the best linear unbiased estimator (BLUE), achieving the smallest variance among all linear unbiased estimators. It should be mentioned that in practice, the true covariance $P_{z}$ is often unknown and should be estimated before applying WLS.

	\subsection{Optimal Realization is a Weighted Least-Squares Solution} \label{Sct4.1}
	
	Returning to the realization problem, let $\hat {\bm{g}}_n$ be the estimate of the first $n$ Markov parameters
	\begin{equation} \label{E37}
		\bm{g}_n := \begin{bmatrix}
			g_{1}&g_2& \cdots &g_n\end{bmatrix},
	\end{equation}
	and its asymptotic distribution is assumed to be
	\begin{equation} \label{E38}
		\sqrt{N}\left(\hat {\bm{g}}_n - \bm{g}_n\right)\sim \text{As}\mathcal{N}\left(0,P_g\right),
	\end{equation}
    where $N$ denotes the number of samples used to identify the Markov parameters, and  $\text{As}\mathcal{N}\left(0,P_g\right)$ refers to an asymptotic normal distribution with mean zero and covariance $P_g$. The residual of $\bm{a}{\hat {\mathcal{H}}}_{n_xp}^{+} + \hat {\mathcal{H}}_{n_xp}^{-}$ is
	\begin{equation} \label{E39}
		\begin{split}
			&\bm{a}{\hat {\mathcal{H}}}_{n_xp}^{+} + \hat {\mathcal{H}}_{n_xp}^{-} - \left(\bm{a}{{\mathcal{H}}}_{n_xp}^{+} + {\mathcal{H}}_{n_xp}^{-}\right) \\
			&= (\hat {\mathcal{H}}_{n_xp}^{-} - {\mathcal{H}}_{n_xp}^{-}) + \bm{a} (\hat {{\mathcal{H}}}_{n_xp}^{+} - {{\mathcal{H}}}_{n_xp}^{+}) \\
			&= \begin{bmatrix}\bm{a}&1\end{bmatrix} (\hat {{\mathcal{H}}}_{n_xp} - {{\mathcal{H}}}_{n_xp}).
		\end{split}
	\end{equation}
	Since $(\hat {{\mathcal{H}}}_{n_xp} - {{\mathcal{H}}}_{n_xp})$ is a Hankel matrix, we rewrite \eqref{E39} as
	\begin{equation} \label{E40}
		\begin{bmatrix}\bm{a}&1\end{bmatrix}  (\hat {{\mathcal{H}}}_{n_xp} - {{\mathcal{H}}}_{n_xp}) =  (\hat {\bm{g}}_n - {\bm{g}_n}){\mathcal{T}}(\bm{a}),
	\end{equation}
	where  $\mathcal{T}(\bm{a}) \in \mathbb{R}^{n\times(n-n_x)}$ is a Toeplitz matrix with compatible dimension, having $\begin{bmatrix}a_{n_x}&a_{n_x-1}&\cdots&1&0&\cdots &0\end{bmatrix}^\top$ as its first column and $\begin{bmatrix}a_{n_x}&0&\cdots&0\end{bmatrix}$ as its first row. Based on \eqref{E38} and \eqref{E40}, we further conclude that the asymptotic distribution of the residual $(\hat {\bm{g}}_n - {\bm{g}_n}){\mathcal{T}}(\bm{a})$ is
	\begin{equation} \label{E41}
		\sqrt{N}(\hat {\bm{g}}_n - {\bm{g}_n}){\mathcal{T}}(\bm{a}) \sim \text{As}\mathcal{N}\left(0,{\mathcal{T}^\top}(\bm{a})P_{g}{\mathcal{T}}(\bm{a}) \right).
	\end{equation}
	Taking $W(\bm{a}) = \left({\mathcal{T}^\top}(\bm{a})P_g{\mathcal{T}}(\bm{a})\right)^{-1}$ as the optimal weighting for WLS, we obtain the optimal estimate 
	\begin{equation} \label{E42}
		\hat{\bm{a}}_{\text{wls}} = -{\hat {\mathcal{H}}}_{n_xn}^{-}W(\bm{a})({\hat {\mathcal{H}}_{n_xn}^{+}})^{\top}\left({\hat {\mathcal{H}}_{n_xn}^{+}} W(\bm{a})({\hat {\mathcal{H}}_{n_xn}^{+}})^{\top}\right)^{-1}.
	\end{equation}
	
	Although the optimal weighting $W(\bm{a})$ depends on the true value of $\bm{a}$, as demonstrated in \cite{Galrinho2018parametric}, replacing $\bm{a}$ in $W(\bm{a})$ with its estimate  $\hat{\bm{a}}_{\text{ols}}$ or $\hat{\bm{a}}_{\text{tls}}$ will not affect the asymptotic optimality of $\hat{\bm{a}}_{\text{wls}}$. This results in a two-step least squares method: the first step obtains an initial solution using TLS or OLS, and the second step refines this initial estimate using WLS. Moreover, it is possible to continue iterating this procedure, which may improve the estimate for finite samples. 
	
	\textbf{Experiment 4}: To demonstrate the performance of WLS w.r.t. OLS and TLS, here we use the two systems in Experiment~1 to perform simulations. For System~1, we already know that OLS gives better performance than TLS, while for System~2, TLS gives better performance than OLS. To simplify the comparison, we evaluate WLS against OLS for System~1 and against TLS for System~2. The results are presented in Figure~\ref{F4}. As shown from the figure, WLS performs competitively with the best method for each system.
	
	\begin{figure}
		\centering
		\includegraphics[scale=0.5]{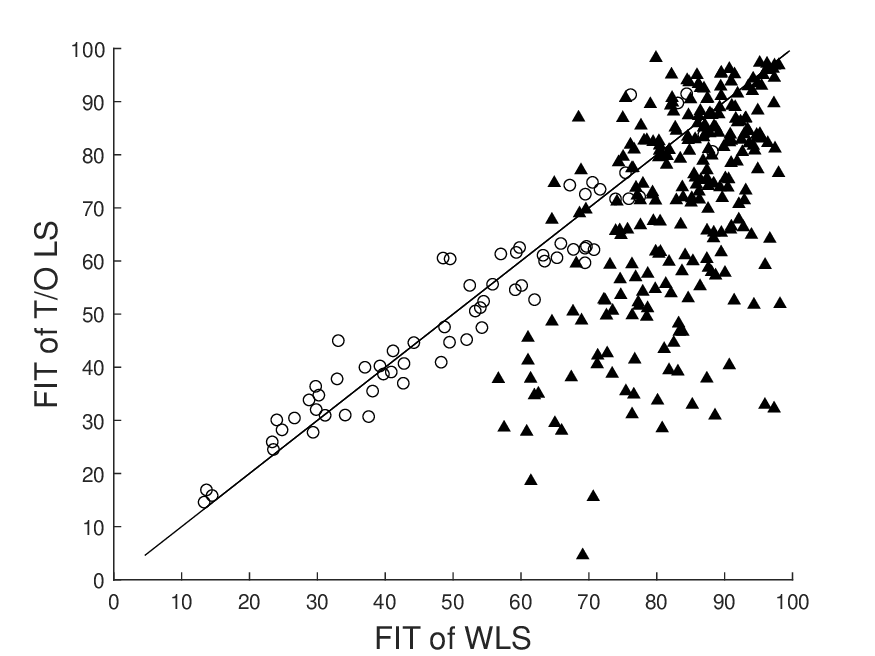}
		\caption{FITs of impulse responses from 200 Monte Carlo trials: System 1 ($\circ$, OLS VS WLS) and System 2 ($\blacktriangle$, TLS VS WLS), and the solid line is a bisector line.}
		\label{F4}
	\end{figure}
	
	\section{Statistical Analysis} \label{Sct5}
	
	In this section, we provide a statistical analysis for the three least-squares algorithms. We start with the estimates of the Markov parameters. Suppose that the first $n$ Markov parameters $\bm{g}_n$ in \eqref{E37} are estimated using $N$ data samples. We make the following assumptions.
	\begin{assumption}  \label{Asp4}
		We let the number of Markov parameters $n$ depends on the sample size $N$ according to the following conditions\footnote{In this assumption, $n$ is denoted by $n(N)$ to highlight the dependency of $n$ on $N$, whereas for simplicity, such a dependence is concealed in other parts of the paper.}:
		
		(1) $n(N) \to \infty$ as $N \to \infty$.
		
		(2) $n^{4+\delta}(N)/N \to 0$ for some $\delta>0$, as $N \to \infty$.
		
		(3) $\sqrt{N}d(N) \to 0$ as $N \to \infty$, where $d(N):= \sum_{k=n(N)+1}^{\infty} \norm{C{A^{k- 1}}B}$.	
	\end{assumption}
	
	\begin{assumption}  \label{Asp2}
		The estimates of the Markov parameters are consistent, i.e.,
		\begin{equation} \label{E43}
			{\hat {\bm{g}}}_n - \bm{g}_n \to 0, {\rm{as}} \ N \to \infty \ {\rm{w.p.1}},
		\end{equation}
		where $\text{w.p.1}$ is an abbreviation for 'with probability one'.
	\end{assumption}
	
	\begin{assumption}  \label{Asp3}
		The estimates of the Markov parameters are asymptotically normal, i.e.,
		\begin{equation} \label{E44}
			\sqrt{N}\left({\hat {\bm{g}}}_n - \bm{g}_n\right) \sim \text{As}\mathcal{N}\left(0,P_g\right).
		\end{equation}
	\end{assumption}

	\begin{remark} \label{Rmk51}
		Assumption \ref{Asp4} ensures that the number of Markov parameters, $n(N)$, grows at a suitable rate with $N$, so that these parameters constitute an asymptotically sufficient statistic. In particular, the first condition ensures that the growth of $n(N)$ is not too slow, while the second condition ensures that the growth of $n(N)$ is not too fast. In principle, one can take $n=\beta\text{log}N$, where $\beta>0$, to satisfy these two conditions for sufficiently large $N$. Moreover, for the third condition, since $\rho(A)<1$, we have $\norm{A^{n(N)}} = \mathcal{O}(\rho^{n(N)}) = \mathcal{O}(N^{-\beta/{\rm{log}(1/\rho)}})$, where $\rho(A) < \rho <1$. In this way, the third condition will be satisfied for a large enough $\beta$. In practice though, $n(N)$ can be determined by minimizing the prediction errors of the estimated state-space model as proposed in \cite{Galrinho2018parametric} for other models estimated with the so-called WNSF approach.
	\end{remark}
	
	\begin{remark} \label{Rmk5}
		The Markov parameters are typically estimated using least-squares methods applied to an FIR or ARX model. It is well known that if the model order tends to infinity at the rate suggested in Assumption \ref{Asp4}, the impulse responses estimates are consistent and asymptotically normal distributed \cite{Ljung1992asymptotic}. Therefore, Assumptions \ref{Asp2} and \ref{Asp3} are reasonable and not overly restrictive.
	\end{remark}
	
	\subsection{Consistency} \label{Sct5.1}
	
	The consistency of the three least-squares solutions is presented in the following theorem:
	\begin{theorem} \label{Thm2}
		Under Assumptions \ref{Asp1}, \ref{Asp4} and \ref{Asp2}, the estimates $\hat{\bm{a}}_{\text{ols}}$, $\hat{\bm{a}}_{\text{tls}}$ and $\hat{\bm{a}}_{\text{wls}}$ are consistent, i.e.,
		\begin{subequations} \label{E45}
			\begin{align}
				\hat{\bm{a}}_{\text{ols}} - \bm{a} &\to 0, {\rm{as}} \ N \to \infty \ {\rm{w.p.1}}, \\
				\hat{\bm{a}}_{\text{tls}} - \bm{a} &\to 0, {\rm{as}} \ N \to \infty \ {\rm{w.p.1}}, \\
				\hat{\bm{a}}_{\text{wls}} - \bm{a} &\to 0, {\rm{as}} \ N \to \infty \ {\rm{w.p.1}}.
			\end{align}
		\end{subequations}
	\end{theorem}
	\begin{proof}
		The proof for the consistency of $\hat{\bm{a}}_{\text{ols}}$ and $\hat{\bm{a}}_{\text{wls}}$ is equivalent to the statistical analysis in \cite{Galrinho2018parametric}, where the Markov parameters are estimated using OLS applied to a high order ARX model, such that Assumption \ref{Asp2} is satisfied. For the consistency of $\hat{\bm{a}}_{\text{tls}}$, see Appendix \ref{App3}.
	\end{proof}
	
	\subsection{Asymptotic Normality} \label{Sct5.2}
	
	The asymptotic normality of the three least-squares solutions is presented in the following theorem:
	\begin{theorem} \label{Thm3}
		Under Assumptions \ref{Asp1}, \ref{Asp4}, \ref{Asp2} and \ref{Asp3}, the estimates $\hat{\bm{a}}_{\text{ols}}$, $\hat{\bm{a}}_{\text{tls}}$ and $\hat{\bm{a}}_{\text{wls}}$ are asymptotically normal, i.e.,
		\begin{subequations} \label{E46}
			\begin{align}
				\sqrt{N}\left(\hat{\bm{a}}_{\text{ols}} - \bm{a}\right) &\sim \text{As}\mathcal{N}\left(0,P_{a,\text{ols}}\right), \\
				\sqrt{N}\left(\hat{\bm{a}}_{\text{tls}} - \bm{a}\right) &\sim \text{As}\mathcal{N}\left(0,P_{a,\text{tls}}\right), \\
				\sqrt{N}\left(\hat{\bm{a}}_{\text{wls}} - \bm{a}\right) &\sim \text{As}\mathcal{N}\left(0,P_{a,\text{wls}}\right),
			\end{align}
		\end{subequations}
		where 
		\begin{subequations} \label{E47}
			\begin{align}
				P_{a,\text{ols}} =& P_{a,\text{tls}} = ({{\mathcal{H}}}_{n_xp}^{+})^{\dagger}(W(\bm{a}))^{-1}(({{\mathcal{H}}}_{n_xp}^{+})^{\dagger})^{\top}, \\
				P_{a,\text{wls}} =&\left({\mathcal{H}}_{n_xp}^{+}{W(\bm{a})}({\mathcal{H}}_{n_xp}^{+})^{\top}\right)^{-1}.
			\end{align}
		\end{subequations}
		Moreover, we have that
		\begin{equation} \label{E48}
			P_{a,\text{ols}}=P_{a,\text{tls}} \succcurlyeq  P_{a,\text{wls}}.
		\end{equation}
	\end{theorem}

	\begin{proof}
		The proof for the asymptotic normality of $\hat{\bm{a}}_{\text{ols}}$ and $\hat{\bm{a}}_{\text{wls}}$ is equivalent to the statistical analysis in \cite{Galrinho2018parametric}. For the asymptotic normality of $\hat{\bm{a}}_{\text{tls}}$, see Appendix \ref{App3}.    
	\end{proof}
	
	\begin{remark}
		Although Theorems \ref{Thm2} and \ref{Thm3} indicate that under Assumptions \ref{Asp4},  \ref{Asp2} and \ref{Asp3}, $\hat{\bm{a}}_{\text{ols}}$ and $\hat{\bm{a}}_{\text{tls}}$ are asymptotically equivalent, their performance is generally different for finite sample settings, as discussed earlier.
	\end{remark}
	
	\begin{remark}
		Theorem \ref{Thm3} shows that WLS yields a smaller variance than OLS and TLS, which sheds new light on pursuing an asymptotically efficient SIM, a long-standing open problem in system identification. To be specific, as a bi-level problem, if the second-level estimate $\hat{\bm{a}}_{\text{wls}}$ is to be asymptotically efficient, then the first-level estimate $\hat{\bm{g}}_n$ must be a sufficient statistic, at least as the sample size grows. This implies that achieving an asymptotically efficient realization requires the estimation of as many Markov parameters as possible in an efficient manner to ensure no loss of information. Besides the pre-estimation of certain Markov parameters, for the second-level estimate $\hat{\bm{a}}_{\text{wls}}$, the optimal weighting matrix $W(\bm{a})$ depends on the true value of $\bm{a}$. In practice, $\bm{a}$ in $W(\bm{a})$ can be replaced by its consistent estimate  $\hat{\bm{a}}_{\text{ols}}$ or $\hat{\bm{a}}_{\text{tls}}$, which will not affect the asymptotic optimality of $\hat{\bm{a}}_{\text{wls}}$.
	\end{remark}

	\section{Conclusion} \label{Sct6}
	
	This work revisits the classical approximate realization problem, where two prototypes of realization algorithms are compared within a least-squares framework. We show that the range-space-based method corresponds to a TLS solution, whereas the null-space-based method corresponds to an OLS solution. By examining the differences and sensitivities of these algorithms, we identify the conditions under which one method may be preferred over the other. Recognizing the suboptimality of both methods, we further propose that the optimal realization is a WLS solution. Moreover, statistical properties of these methods are presented, revealing that the TLS and OLS solutions are asymptotically equivalent, and the WLS solution gives a smaller variance. 
	
	Our work brings new insights into the classical approximate realization problem. First, our results reveal that for a well-conditioned realization problem, the better accuracy of TLS as opposed to OLS will be more pronounced when the ratio $\hat\sigma_{n_x}/\hat\sigma_{n_x}^{+}$ is larger. Second, we highlight several scenarios when range-space-based methods are likely to have a sensitive solution when the difference $\hat\sigma_{n_x}^{+}-\hat\sigma_{n_x+1}$ is small. Such situations may arise, for example, if matrix $A$ has a spectral radius $\rho(A)$ close to zero or if the system is of high order. Under these circumstances, the null-space-based methods might be a better option. Moreover, we point out that to achieve an asymptotically efficient realization method, the number of Markov parameters estimated in the first step should grows at a suitable rate with the sample size, serving as an asymptotically sufficient statistic. Then, WLS should be used in the second step, where the optimal weighting depends on the initial estimate and its covariance.
	
	In this work, we focused on SISO systems to simplify the analysis. Extending the WLS realization method to multiple-input single-output systems is straightforward, as they share the same parameterization for the left null space of the extended observability matrix as SISO systems. Furthermore, based on our recent work \cite{He2024weighted}, when the primary interest is in the $A$-matrix, a similar extension to multiple-output systems can be accomplished by building a block Hankel matrix. Future research will explore the application of these results to subspace identification.

	\section*{References}
	\bibliographystyle{IEEEtran}
	\bibliography{refs}

	\appendices 
	\numberwithin{equation}{section}
	
	\section{Proof of Theorem \ref{Thm0}} \label{App0}
	Notice that the extended observability matrix $\hat{\mathcal{O}}_{n_x}$ in \eqref{E31} satisfies the so-called backward shift property \cite{De2020least}:
	    \begin{equation} \label{AE1}
			\hat{\mathcal{O}}_{n_x}^{-}= \hat{\mathcal{O}}_{n_x}^{+} \hat A.
		\end{equation}
		Meanwhile, the characteristic coefficients $\left\{{\hat a}_i\right\}_{i=1}^{n_x}$ of $\hat A$ satisfy
		\begin{equation}  \label{AE2}
			\begin{bmatrix}
				{\hat a}_{n_x}&{\hat a}_{{n_x}-1}& \cdots &{\hat a}_1&1\end{bmatrix} {\hat{\mathcal{O}}}_{n_x} = 0.
		\end{equation}
		The above two equations suggest that the range space of the extended observability matrix ${\hat{\mathcal{O}}}_{n_x}$ is solely determined by the eigenvalues and their multiplicity structure, which simultaneously parameterize the null space of ${\hat{\mathcal{O}}}_{n_x}$. Since the singular vectors are orthogonal to each other, it is clear that the solution \eqref{E29} satisfies $\begin{bmatrix}\hat{\bm{a}}_\text{tls}&1\end{bmatrix}{\hat{\mathcal{O}}}_{n_x} = 0$. Moreover, as $\text{dim}\left(\mathcal{K}({\hat{\mathcal{O}}}_{n_x}^\top)\right)=1$, the solution to \eqref{AE2} is unique, so we conclude that the solution \eqref{E31} to the equation \eqref{AE1} is the same as the the solution \eqref{E29} to the equation \eqref{AE2}, up to a similarity transformation. \hfill$\blacksquare$
		
	\section{Proof of Lemma \ref{Lem0}} \label{App1}
	The first equality in \eqref{E34d1} comes from Definition~\ref{Def1}, where the vectors $\begin{bmatrix}\hat{\bm{a}}_\text{tls} & 1\end{bmatrix}$ and $\begin{bmatrix}\bm{a} & 1\end{bmatrix}$ uniquely span the subspaces $\mathcal{K}\left(\doublehat{\mathcal{H}}_{n_xp}^\top\right)$ and $\mathcal{K}\left(\mathcal{H}_{n_xp}^\top\right)$, respectively. The second equality in \eqref{E34d1} comes from the property that the distance defined in Definition~\ref{Def2} is invariant under the operation of taking orthogonal complements. Moreover, the last equality in \eqref{E34d1} comes from inequality (12) in \cite{Van1989accuracy}, which applies the generalized $\sin\theta$ theory \cite{Wedin1972perturbation} to $\hat{\mathcal{H}}_{n_xp}$ and ${\mathcal{H}}_{n_xp}$.
	The proof of \eqref{E34d2} is identical to that of \eqref{E34d1}. \hfill$\blacksquare$

	\section{Proof of Lemma \ref{Lem3}} \label{App2}
	Based on interlacing inequalities for singular values \cite{Golub2013matrix}, we have that for $i=1,2,...,n_x$,
		\begin{equation*}
			\hat{\sigma}_i \geq {\hat\sigma}_{i}^{+} \geq \hat{\sigma}_{i+1}.
		\end{equation*}
		Letting $i=n_x$, we have that ${\hat\sigma}_{n_x}^{+} - \hat\sigma_{n_x+1} \geq 0$. For the upper bound, since $\hat\sigma_{n_x+1}\geq 0$, we have that ${\hat\sigma}_{n_x}^{+} - \hat\sigma_{n_x+1} \leq {\hat\sigma}_{n_x}^{+}$. We now provide a upper bound for ${\hat\sigma}_{n_x}^{+}$, the smallest singular value of the Hankel matrix $\hat{\mathcal{H}}_{n_xp}^{+}$. It is straightforward to see that
		\begin{equation}  \label{BE1}
			\begin{split}
			   {\hat\sigma}_{n_x}^{+} 	&\leq \norm{\hat{\mathcal{H}}_{n_xp}^{+}}  \leq 
				\norm{{\mathcal{H}}_{n_xp}^{+}} + \norm{\tilde{\mathcal{H}}_{n_xp}^{+}} \\ &= \norm{\mathcal{O}_{n_x-1}\mathcal{C}_p} + \norm{\tilde{\mathcal{H}}_{n_xp}^{+}} \\ 
				&\leq  \norm{\mathcal{O}_{n_x-1}}\norm{\mathcal{C}_p} + \norm{\tilde{\mathcal{H}}_{n_xp}^{+}}.
			\end{split}	
		\end{equation}
		Moreover, using the norm inequality between a block matrix and its submatrices, $\norm{\mathcal{O}_{n_x-1}}$ and $\norm{\mathcal{C}_p}$ are bounded separately by
		\begin{subequations}  \label{BE2}
			\begin{align}
				\norm{\mathcal{O}_{n_x-1}} &\leq \sqrt{\sum_{k=0}^{n_x-1}\norm{CA^k}^2}\leq \norm{C}\sqrt{\sum_{k=0}^{n_x-1}\norm{A^k}^2},  \label{BE2a}\\
				\norm{\mathcal{C}_p} &\leq \sqrt{\sum_{k=0}^{p}\norm{BA^k}^2}\leq \norm{B}\sqrt{\sum_{k=0}^{p}\norm{A^k}^2}  \label{BE2b}.
			\end{align}	
		\end{subequations}
		Furthermore, utilizing the ratio between the exponents of the spectral norm and the square root of the spectral radius defined by $\Phi(A)$ \cite{Oymak2021revisiting}, we have that
		\begin{equation}  \label{BE3}
			\norm{A^k} \leq \Phi(A)\rho(A)^{k/2}.
		\end{equation}
		Substituting \eqref{BE3} into \eqref{BE2a} and \eqref{BE2b}, we have that
		\begin{subequations}  \label{BE4}
			\begin{align}
				\norm{\mathcal{O}_{n_x-1}} &\leq \norm{C}\Phi(A)
				\sqrt{\frac{1-\rho(A)^{n_x}}{{1-\rho(A)}}},   \label{BE4a}\\
				\norm{\mathcal{C}_p} &\leq \norm{B}\Phi(A)
				\sqrt{\frac{1-\rho(A)^{p+1}}{{1-\rho(A)}}}  \label{BE4b}.
			\end{align}	
		\end{subequations}
		After merging \eqref{BE4a} and \eqref{BE4b} together, we have
		\begin{equation*}
			\norm{\mathcal{O}_{n_x-1}}\norm{\mathcal{C}_p} \leq  \norm{C}\norm{B}\Phi(A)^2\frac{1-\rho(A)^{n-n_x}}{{1-\rho(A)}}.
		\end{equation*}
		Substituting the above inequality into \eqref{BE1}, the proof is completed. \hfill$\blacksquare$

	\section{Statistical Analysis}   \label{App3}
	
	\subsection{Auxiliary Results}
	To prove Theorems \ref{Thm2} and \ref{Thm3} for TLS, we introduce some auxiliary results.
	
	(a) $\norm{{\hat {\mathcal{H}}}_{n_xp} - {\mathcal{H}}_{n_xp}}  \to 0$, ${\rm{as}} \ N \to \infty$ $\text{w.p.1}$: Using the norm inequality of a block matrix in Lemma \ref{LemG1}, we have that
	\begin{equation} \label{CE1}
		\norm{{\hat {\mathcal{H}}}_{n_xp} - {\mathcal{H}}_{n_xp}} \leq \sqrt{n_x+1} \norm{\hat{\bm{g}}_n - \bm{g}_n}.
	\end{equation}
	Under Assumption \ref{Asp2}, $\norm{\hat{\bm{g}}_n - \bm{g}_n} \to 0$, as $N \to \infty \ {\text{w.p.1}}$, so we conclude that $\norm{{\hat {\mathcal{H}}}_{n_xp} - {\mathcal{H}}_{n_xp}}  \to 0$, ${\rm{as}} \ N \to \infty$ $\text{w.p.1}$. Moreover, since ${\mathcal{H}}_{n_xp}^{-}$ and ${\mathcal{H}}_{n_xp}^{+}$ are submatrices of ${\mathcal{H}}_{n_xp}$, we have that
	\begin{subequations} \label{CE2}
		\begin{align}
			&\norm{ {\hat {\mathcal{H}}}_{n_xn}^{-} - {\mathcal{H}}_{n_xn}^{-} } \to 0, {\rm{as}} \ N \to \infty \ {\text{w.p.1}}, \label{CE2a}\\
			&\norm{ {\hat {\mathcal{H}}}_{n_xn}^{+} - {\mathcal{H}}_{n_xn}^{+} }\to 0, {\rm{as}} \ N \to \infty \ {\text{w.p.1}} \label{CE2b}.
		\end{align}
	\end{subequations}
	
	(b) $\hat{\sigma}_{n_x+1} \to 0$, ${\rm{as}} \ N \to \infty$ $\text{w.p.1}$:
	According to the auxiliary result (a), we have $\norm{{\hat {\mathcal{H}}}_{n_xp} - {\mathcal{H}}_{n_xp}}  \to 0$, ${\rm{as}} \ N \to \infty$ $\text{w.p.1}$. Since ${\mathcal{H}}_{n_xp}$ is low-rank, using the continuity of eigenvalues, we conclude that $\hat{\sigma}_{n_x+1} \to 0$, ${\rm{as}} \ N \to \infty$ $\text{w.p.1}$.
	
	(c) $\norm{{\mathcal{H}}_{n_xp}}$ is bounded for all $n$: Similarly, using the norm inequality of a block matrix in Lemma \ref{LemG1}, we have that
	\begin{equation} \label{CE3}
		\norm{{\mathcal{H}}_{n_xp}} \leq \sqrt{n_x+1} \norm{\bm{g}_n}, \forall n.
	\end{equation}
	Under Assumption \ref{Asp1}, the system is asymptotically stable, which means that the Markov parameters $\left\{g_i = C{A^{i - 1}}B\right\}$ are exponentially decaying with $i$. This ensures that $\norm{\bm{g}_n}$ is bounded for all $n$, which further implies that $\norm{{\mathcal{H}}_{n_xp}}$ is bounded for all $n$.
	
	(d) ${\mathcal{T}}(\bm{a})$ is bounded for all $n$: See \cite[Th. 1]{Galrinho2018parametric}.
	
	\subsection{Proof of Theorem \ref{Thm2} for TLS (Consistency)}
	Define $\tilde{\bm{a}}_{\text{tls}} := \hat{\bm{a}}_{\text{tls}} - \bm{a}$ and $\tilde{\bm{g}}_{n} := \hat{\bm{g}}_{n} - \bm{g}_{n}$. We first have that
		\begin{equation} \label{CE4} 
			\begin{split}
				\hat{\bm{a}}_\text{tls} - \bm{a} =& \left(-\hat {\mathcal{H}}_{n_xp}^{-}(\hat {\mathcal{H}}_{n_xp}^{+})^\top-
				\bm{a}\hat {\mathcal{H}}_{n_xp}^{+}(\hat {\mathcal{H}}_{n_xp}^{+})^{\top}+\bm{a}\hat \sigma_{n_x+1}^2I\right)
				\\
				&\times\left(\hat {\mathcal{H}}_{n_xp}^{+}(\hat {\mathcal{H}}_{n_xp}^{+})^{\top}-\hat \sigma_{n_x+1}^2I\right)^{-1} \\
				= & -\tilde{{\bm{g}}}_n{\mathcal{T}}(\bm{a})({ \hat{\mathcal{H}}}_{n_xp}^{+})^{\top}\left(\hat{\mathcal{H}}_{n_xp}^{+}(\hat{\mathcal{H}}_{n_xp}^{+})^{\top} - {\hat\sigma}_{n_x+1}^2I\right)^{-1}
				\\
				&  +\bm{a}{\hat\sigma}_{n_x+1}^2\left(\hat{\mathcal{H}}_{n_xp}^{+}(\hat{\mathcal{H}}_{n_xp}^{+})^{\top} - {\hat\sigma}_{n_x+1}^2I\right)^{-1}.
			\end{split}
		\end{equation} 
		Using the triangular inequality, we have that
		\begin{equation} \label{CE4a} 
			\begin{split}
				\norm{\tilde{\bm{a}}_{\text{tls}}} \leq  &\norm{\tilde{{\bm{g}}}_n}\norm{{\mathcal{T}}(\bm{a})}\norm{{ \hat{\mathcal{H}}}_{n_xp}^{+}}\times \\
				&\norm{\left(\hat{\mathcal{H}}_{n_xp}^{+}(\hat{\mathcal{H}}_{n_xp}^{+})^{\top} - {\hat\sigma}_{n_x+1}^2I\right)^{-1}} +\norm{\bm{a}}\times 
				\\
				&\norm{{\hat\sigma}_{n_x+1}^2}\norm{\left(\hat{\mathcal{H}}_{n_xp}^{+}(\hat{\mathcal{H}}_{n_xp}^{+})^{\top} - {\hat\sigma}_{n_x+1}^2I\right)^{-1}}.
			\end{split}
		\end{equation}
		According to the auxiliary result (b), ${\hat\sigma}_{n_x+1}^2 \to 0$, ${\rm{as}} \ N \to \infty$ $\text{w.p.1}$. Moreover, according to the auxiliary results (a) and (c), $\norm{{\hat{\mathcal{H}}}_{n_xp}^{+} - {\mathcal{H}}_{n_xp}^{+}}  \to 0$, ${\rm{as}} \ N \to \infty$ $\text{w.p.1}$, and $\norm{{\mathcal{H}}_{n_xp}^{+}}$ is bounded. In this way, using Lemma \ref{LemG2}, we have that
		\begin{equation} \label{CE5}
			\norm{\hat{\mathcal{H}}_{n_xp}^{+}(\hat{\mathcal{H}}_{n_xp}^{+})^{\top} - {\hat\sigma}_{n_x+1}^2I - {\mathcal{H}}_{n_xp}^{+}({\mathcal{H}}_{n_xp}^{+})^{\top}} \to 0, 
		\end{equation}
		${\rm{as}} \ N \to \infty \ \text{w.p.1}$. Using Lemma \ref{LemG3}, we further have that
		\begin{equation} \label{CE6}
			\norm{\left(\hat{\mathcal{H}}_{n_xp}^{+}(\hat{\mathcal{H}}_{n_xp}^{+})^{\top} - {\hat\sigma}_{n_x+1}^2I\right)^{-1} - \left({\mathcal{H}}_{n_xp}^{+}({\mathcal{H}}_{n_xp}^{+})^{\top}\right)^{-1}}  \to 0,
		\end{equation}
		${\rm{as}} \ N \to \infty$, $\text{w.p.1}$. Moreover, auxiliary results (c) and (d) show that $\norm{{\mathcal{H}}_{n_xp}^{+}}$ and  $\norm{{\mathcal{T}}(\bm{a})}$ are bounded, and Assumption \ref{Asp2} implies that $\tilde{\bm{g}}_{n} \to 0$, ${\rm{as}} \ N \to \infty$ $\text{w.p.1}$. Putting them together, and combining \eqref{CE4a} and \eqref{CE6}, we therefore conclude that 
		\begin{equation}  \label{CE7}
			\tilde{\bm{a}}_{\text{tls}} \to 0, {\rm{as}} \ N \to \infty \ {\text{w.p.1}}.
		\end{equation}
	    The proof is completed. \hfill$\blacksquare$

	\subsection{Proof of Theorem \ref{Thm3} for TLS (Asymptotic Normality)}
	We first study the covariance of $\tilde{\bm{a}}_{\text{tls}} := \hat{\bm{a}}_{\text{tls}} - \bm{a}$ given in \eqref{CE4}. Under Assumption \ref{Asp3}, we have that $\sqrt{N}\text{Vec}\left({\hat {\mathcal{H}}}_{n_xp} - \mathcal{H}_{n_xp} \right)$ is asymptotically normally distributed, where $\text{Vec}(\cdot)$ means the vectorization of a matrix. According to Lemma \ref{LemG6}, we have that $\sqrt{N}{\hat\sigma}_{n_x+1}$ is 
		asymptotically normally distributed, which means that ${\hat\sigma}_{n_x+1}^2$ is $\mathcal{O}(1/N)$. Hence,  $\sqrt{N}{\hat\sigma}_{n_x+1}^2$ is dominated by the asymptotically normally distributed term $\sqrt{N}\tilde{{\bm{g}}}_T$. In this way, $\tilde{\bm{a}}_{\text{tls}}$ in \eqref{CE4} can be rewritten as
		\begin{equation*}
			\tilde{\bm{a}}_{\text{tls}} = -\tilde{{\bm{g}}}_n{\mathcal{T}}(\bm{a})({ \hat{\mathcal{H}}}_{n_xp}^{+})^{\top}\left(\hat{\mathcal{H}}_{n_xp}^{+}(\hat{\mathcal{H}}_{n_xp}^{+})^{\top}\right)^{-1} + \mathcal{O}(1/N).	
		\end{equation*}
		Using standard results in asymptotic analysis, the faster term $\mathcal{O}(1/N)$ can be neglected when studying the asymptotic distribution. Therefore, it is straightforward to see that $\tilde{\bm{a}}_{\text{tls}}$ has the same asymptotic distribution as the term $\tilde{{\bm{g}}}_n{\mathcal{T}}(\bm{a})({ \hat{\mathcal{H}}}_{n_xp}^{+})^{\top}\left(\hat{\mathcal{H}}_{n_xp}^{+}(\hat{\mathcal{H}}_{n_xp}^{+})^{\top}\right)^{-1}$, whose variance is shown to be $P_{a,\text{ols}} = ({{\mathcal{H}}}_{n_xp}^{+})^{\dagger}W^{-1}(\bm{a})(({{\mathcal{H}}}_{n_xp}^{+})^{\dagger})^{\top}$ \cite{Galrinho2018parametric}. 
		
		We now show that $P_{a,\text{ols}} \succcurlyeq  P_{a,\text{wls}}$.
		First, we have that
		\begin{equation*}
			P_{a,\text{ols}} - P_{a,\text{wls}} = \left({\mathcal{H}}_{n_xp}^{+}({\mathcal{H}}_{n_xp}^{+})^{\top}\right)^{-1}\Gamma \left({\mathcal{H}}_{n_xp}^{+}({\mathcal{H}}_{n_xp}^{+})^{\top}\right)^{-1},
		\end{equation*}
		where 
		\begin{equation*} \label{BE5}
			\begin{split}
				\Gamma =& {{\mathcal{H}}}_{n_xp}^{+}W(\bm{a})^{-1}({{\mathcal{H}}}_{n_xp}^{+})^{\top} - {\mathcal{H}}_{n_xp}^{+}({\mathcal{H}}_{n_xp}^{+})^{\top}\times\\
				&\left({\mathcal{H}}_{n_xp}^{+}{W(\bm{a})}({\mathcal{H}}_{n_xp}^{+})^{\top}\right)^{-1}{\mathcal{H}}_{n_xp}^{+}({\mathcal{H}}_{n_xp}^{+})^{\top}.
			\end{split}	
		\end{equation*}
		According to the Shur complement, $\Gamma \succcurlyeq 0$ is equivalent to 
		\begin{equation*}
			\begin{split}
				&\begin{bmatrix}
					{{\mathcal{H}}}_{n_xp}^{+}W(\bm{a})^{-1}({{\mathcal{H}}}_{n_xp}^{+})^{\top} &{\mathcal{H}}_{n_xp}^{+}({\mathcal{H}}_{n_xp}^{+})^{\top}\\
					{\mathcal{H}}_{n_xp}^{+}({\mathcal{H}}_{n_xp}^{+})^{\top}&{\mathcal{H}}_{n_xp}^{+}{W(\bm{a})}({\mathcal{H}}_{n_xp}^{+})^{\top}
				\end{bmatrix}\\
				& = \begin{bmatrix}
					{{\mathcal{H}}}_{n_xp}^{+}&0\\
					0&{{\mathcal{H}}}_{n_xp}^{+}
				\end{bmatrix}\begin{bmatrix}
					W(\bm{a})^{-1} &I\\
					I&W(\bm{a})
				\end{bmatrix}\begin{bmatrix}
					({\mathcal{H}}_{n_xp}^{+}) &0\\
					0&({\mathcal{H}}_{n_xp}^{+})
				\end{bmatrix}^{\top} \\
				&\succcurlyeq 0.
			\end{split}
		\end{equation*}
		Since $W(\bm{a}) \succcurlyeq 0$, we have $\begin{bmatrix}
			W(\bm{a})^{-1} &I\\
			I&W(\bm{a})
		\end{bmatrix} \succcurlyeq 0$. As a result, we have that $\Gamma \succcurlyeq 0$, which gives
		\begin{equation*}
			P_{a,\text{ols}} \succcurlyeq  P_{a,\text{wls}}.
		\end{equation*}
	The proof is completed. \hfill$\blacksquare$

	\section{Technical Lemmas}   \label{App5}
	
	\begin{lemma} [{\cite[Lemma A.1]{Tsiamis2019finite}}, Norm of a block matrix] \label{LemG1}
		Let $M$ be a block-column matrix defined as $M = \begin{bmatrix}
			M_1^\top&M_2^\top&\cdots &M_f^\top
		\end{bmatrix}^\top$, where all the $M_i$'s have the same dimension. Then, the block matrix $M$ satisfies
		\begin{equation*}
			\norm{M} \leq \sqrt{f} \max\limits_{1\leq i\leq f} \norm{M_i}.
		\end{equation*}
	\end{lemma}
	
	\begin{lemma} [{\cite[Proposition 1]{Galrinho2018parametric}}] \label{LemG2}
		Consider the product $\prod_{i=1}^p \hat{M}_N^{(i)}$, where $p$ is finite and $\hat{M}_N^{(i)}$ are stochastic matrices of appropriate dimensions (possibly functions of $N$) such that
		\begin{equation*}
			\norm{\hat{M}_N^{(i)} - {M}_N^{(i)}} \to 0, {\rm{as}} \ N \to \infty \ {\text{w.p.1}}.
		\end{equation*}
		where ${M}_N^{(i)}$ is a deterministic matrix for each $N$ satisfying $\norm{{M}_N^{(i)}} < c_i$, whose dimensions match the dimensions of $\hat{M}_N^{(i)}$. Then, we have that
		\begin{equation*}
			\norm{\prod_{i=1}^p \hat{M}_N^{(i)} - \prod_{i=1}^p {M}_N^{(i)}} \to 0, {\rm{as}} \ N \to \infty \ {\text{w.p.1}}.
		\end{equation*}
	\end{lemma}
	
	\begin{lemma} [{\cite[Th. 4.1]{Wedin1973perturbation}}] \label{LemG3}
		Consider rank $m$ matrices $M_1\in \mathbb{R}^{m\times n}$ and $M_2\in \mathbb{R}^{m\times n}$, where $m\leq n$. Then, we have 
		\begin{equation}
			\nonumber
			\norm{M_1^\dagger-M_2^\dagger} \leq \sqrt{2}\norm{M_1^\dagger}\norm{M_2^\dagger}\norm{M_1-M_2}.
		\end{equation}
	\end{lemma}
	
	\begin{lemma} [{\cite[Th. 1]{Bura2008distribution}}] \label{LemG6}	
		Let $\hat{\Omega}_T \in \mathbb{R}^{n_z\times q}$ be an estimate of a low rank matrix ${\Omega}_T$ based on a random sample of size $N$. Assume that $\hat{\Omega}_T$ is asymptotically normally distributed as
		\begin{equation}
			\sqrt{N}\text{Vec}\left(\hat{\Omega}_T-{\Omega}_T\right) \sim \text{As}\mathcal{N}\left(0,P_\Omega\right).
		\end{equation}
		Moreover, the SVD Of $\hat{\Omega}_T$ and ${\Omega}_T$ are
		\begin{equation}
			{\Omega}_T = P\begin{bmatrix}\Lambda&0\\0&0\end{bmatrix} Q^\top, \
			\hat{\Omega}_T = \hat P\begin{bmatrix}\hat\Lambda_1&0\\0&\hat\Lambda_0
			\end{bmatrix} \hat Q^\top,
		\end{equation}
		respectively, where $\Lambda = \text{diag}\left(\sigma_1,\sigma_2,\cdots,\sigma_{n+1}\right)$, and $\sigma_1 \geq \sigma_2 \geq \cdots \geq \sigma_{n+1} > 0$.
		Then as $T \to \infty$, we have
		\begin{equation}
			\hat\Lambda_1 \overset{p}{\to} \Lambda, \ \hat\Lambda_0 \overset{p}{\to} 0.
		\end{equation}
		Moreover, $\hat\Lambda_0$ is asymptotically normally distributed as
		\begin{equation}
			\sqrt{N}\text{vec}\left(\hat\Lambda_0\right) \sim \text{As}\mathcal{N}\left(0,P_{\Lambda_0}\right),
		\end{equation}
		where $P_{\Lambda_0} = (P_0^\top\otimes Q_0^\top)P_\Omega(P_0\otimes Q_0)$.
	\end{lemma}
	
\end{document}